\documentclass[twocolumn,pra,showpacs,superscriptaddress]{revtex4-1}

\usepackage{graphicx}
\usepackage{mathtools}
\usepackage{enumerate,amsmath,amssymb,amstext,amsthm}
\usepackage{mathrsfs}

\newtheorem{thm}{Theorem}

\newtheorem{lmm}[thm]{Lemma}

\begin{document}

\newcommand{\mcs}{\mathcal{S}}
\newcommand{\mcp}{\mathcal{P}}
\newcommand{\mcd}{\mathcal{D}}
\newcommand{\mcc}{\mathcal{C}}
\newcommand{\mct}{\mathcal{T}}

\def\be{\begin{equation}}
\def\ee{\end{equation}}
\def\bi{\begin{itemize}}
\def\ei{\end{itemize}}
\newcommand{\la}{\langle}
\newcommand{\ra}{\rangle}
\newcommand{\ta}{\theta}


\newcommand{\R}{\mathbb{R}}
\newcommand{\C}{\mathbb{C}}
\newcommand{\pab}{p(a,b)}
\newcommand{\calo}{\mathcal{O}}
\newcommand{\co}{\mathrm{Cor}}
\newcommand{\beh}{\mathrm{Beh}}
\newcommand{\cut}{\mathrm{Cut}}
\newcommand{\met}{\mathrm{Met}}

\newcommand{\calE}{\mathcal{E}}
\newcommand{\calH}{\mathcal{H}}
\newcommand{\calK}{\mathcal{K}}
\newcommand{\calA}{\mathcal{A}}
\newcommand{\calS}{\mathcal{S}}
\newcommand{\calQ}{\mathcal{Q}}

\newcommand{\states}[1]{\mathscr{S}(#1)}

\def\bi{\begin{itemize}}
\def\ei{\end{itemize}}
\newcommand{\rank}{\textrm{rank}}
\newcommand{\cornm}{\pi(\calE_{n+m})}

\newcommand{\bramatket}[3]{\langle #1 \hspace{1pt} | #2 | \hspace{1pt} #3 \rangle}
\newcommand{\nbox}[2]{\hspace{#2pt} \mbox{#1} \hspace{#2pt}}
\newcommand{\ket}[1]{| \hspace{1pt} #1 \rangle}
\newcommand{\ketL}[1]{| \hspace{1pt} #1 \rangle\rangle}

\newcommand{\ketbrad}[2]{| \hspace{1pt} #1 \rangle \langle #2 \hspace{1pt} |}
\newcommand{\ketbra}[1]{\ketbrad{#1}{#1}}
\newcommand{\braket}[2]{\langle #1 \hspace{1pt} | \hspace{1pt} #2 \rangle}
\newcommand{\bra}[1]{\langle #1 \hspace{1pt} |}
\newcommand{\braL}[1]{\langle\langle #1 \hspace{1pt} |}
\newcommand{\abs}[1]{| #1 |}
\newcommand{\avg}[1]{\left\langle #1 \right\rangle}
\newcommand{\tr}{\mathrm{tr}}
\newcommand{\Tr}{\mathrm{Tr}}
\newcommand{\gell}{\Gamma_\ell}
\newcommand{\ginf}{\Gamma_\infty}
\newcommand{\ext}{\mathrm{ext}}
\newcommand{\hc}{\hat{C}}
\newcommand{\gram}{\mathrm{Gram}}
\newcommand{\Span}{\mathrm{span}}
\newcommand{\Spec}{\mathrm{spec}}

\newcommand{\diff}{\mathrm{d}}

\newcommand{\norm}[1]{\left\lVert#1\right\rVert}

\newcommand{\ip}[2]{\left\langle #1 , #2\right\rangle} 

\newcommand{\id}{\mathds{1}}

\newcommand\inner[2]{\langle #1, #2 \rangle}
\newcommand\innerHS[2]{\langle #1, #2 \rangle_{\mathrm{HS}}}

\newcommand{\otimesmin}{\otimes_{\min}}
\newcommand{\otimesmax}{\otimes_{\max}}

\newcommand{\qset}[1]{\mathcal{Q}(#1)}


\title{Structure  of the set of quantum correlators via semidefinite programming }


\author{Le Phuc Thinh}
\email{cqtlept@nus.edu.sg}
\affiliation{Centre for Quantum Technologies, National University of Singapore}
%
\author{Antonios Varvitsiotis}
\email{avarvits@gmail.com}
\affiliation{Department  of Electrical \& Computer Engineering, National University of Singapore}
\author{Yu Cai}
\email{caiyu01@gmail.com}
\affiliation{Centre for Quantum Technologies, National University of Singapore}


\begin{abstract}
Quantum information leverages 
 properties of quantum behaviors in order to perform useful tasks such as secure communication and randomness certification. Nevertheless, not much is known about the intricate geometric  features of the   set quantum behaviors.
In this paper we study the structure of the set of quantum correlators 
 using  semidefinite programming.
  Our  main results  are $(i)$ a generalization of the analytic description by Tsirelson-Landau-Masanes,     $(ii)$ necessary and sufficient  conditions for    extremality  and exposedness, and  $(iii)$  an operational interpretation of extremality in the case of two dichotomic measurements, in terms of self-testing. 
We illustrate the usefulness of  our theoretical findings  with many  examples and extensive computational work.

\end{abstract}

\maketitle

\section{\label{sec:intro}Introduction}
Quantum theory distinguishes itself from classically-conceived theories of physics in several ways, most notably in terms of the observable correlations  between space-like separated parties. In the seminal work \cite{bell}, John Bell discovered  that quantum theory allows correlations  between space-like separated observers beyond those allowed by local realistic classical theories, a property known as ``nonlocality''. The existence of nonlocal quantum correlations  predicted theoretically by Bell has  been  confirmed by numerous experiments~\cite{hensen, giustina, shalm}, which have influenced  deeply our understanding of the physical world, and have led  to  real-life applications in cryptography~\cite{ekert91} and randomness certification~\cite{colbeck} among~others.

Despite the usefulness of the nonlocal distributions  that  can be generated using quantum resources, the structure of the  set of quantum distributions is not well understood.  Most notably, there are many  semidefinite  programming hierarchies     that approximate the set of quantum correlations from the exterior, but these hierarchies only converge in the limit, e.g., see \cite{npa,wehner}. In fact, it was  shown only recently in \cite{slofstra} that the set of (tensor product)  quantum correlations is not closed. 

Furthermore, it is known that quantum behaviors are not maximally non-local, and from a fundational perspective, there has been an ongoing research program to identify the physical principles behind quantum theory. Several principles have been proposed, such as information causality~\cite{infocausal} and macroscopic locality~\cite{macrolocal}, but none succeeded to single out the  set of quantum distributions. 

From a practical perspective, to manipulate quantum information, one needs a robust method for  identifying quantum systems. This is a challenging task  due to the intrinsic probabilistic nature of quantum theory. Motivated by this, there has been extensive work studying  the relationship between objects in the theory, namely quantum states and measurements, and the theory's predictions, namely probabilities of experiments. This direction 
line of research is known in the literature as quantum tomography whose recent reincarnation is known as self-testing and gate set tomography. 


In this work we study the  set  of quantum correlations by taking  a   geometric viewpoint. Specifically, as the quantum set  is convex,  its properties can be understood via various   features of its convex geometry, such as its facial structure and its  extreme  points \cite{webster}.

 The main technical tool used in this work  is semidefinite programming, whose relevance is already implicit in the work of Tsirelson~\cite{TS87}, a connection that  was pursued further in~\cite{npa,wehner}. Specifically, it essentially follows from  \cite{TS87} that the geometry of the set of quantum correlators is the projection of the geometry of the elliptope, a convex set of central interest in combinatorial optimization~\cite{DL}.

 The link to semidefinite programming   leads  to several interesting results. First, we generalize  the Tsirelson-Landau-Masanes analytic description for the set of quantum correlators, in the case where at least one party performs at most two dichotomic measurements   (Theorem~\ref{thm:main}). Second, we derive     necessary and sufficient conditions for extremality (Theorem~\ref{extpointschar}),  and a sufficient condition for exposedness (Theorem~\ref{thm:exposed}). Third,   we give  an operational interpretation of extremality in the case of two dichotomic measurements (Theorem~\ref{thm:extreme-selftest}). 
 
 Our last result is particularly striking as it highlights that natural geometric features of the set of quantum correlators, such as extremality, correspond to important physical and operational properties.

%
 
 Lastly, we note that the approach of studying the quantum set within the framework  of convex analysis  was  also employed  in \cite{geometry}. Our approach differs in that 
we  use as our main tool  properties of semidefinite programming.

We start  with a brief review of Bell nonlocality (Section~\ref{sec:nonlocality}), followed 
followed by the precise definition of the set of quantum correlators---our main object of study---and a characterization via the positive semidefinite matrix completion problem (Section \ref{sec:linktosdp}), 
 which leads to an analytic description (Section~\ref{sec:analytric}). Further expoiting this connection, 
we study the convex geometry of quantum correlators, its extreme points (Section~\ref{sec:extreme}) and exposed points (Section~\ref{sec:exposed}), and present a connection between the geometric concept of extremality with the operational task of self-testing (Section~\ref{sec:extreme-selftest}). Algorithmic results and examples are interspersed in these sections whenever convenient. We conclude the paper with a high level overview of our results and pointers towards future~research.

\section{\label{sec:nonlocality}Bell nonlocality}
\noindent {\em Bell experiment. } 
The experimental and mathematical framework for studying behaviors between two space-like separated parties  is known as a bipartite Bell experiment. In this setting two parties, called Alice and Bob, perform independently and simultaneously local measurements on their  corresponding subsystems and  record the resulting  outcomes. In  this work, we restrict to Bell experiments where the   parties   can only  perform {\em dichotomic} (i.e.  two-outcome) measurements. 


%

The event  that the first party performed measurement $x$ and got outcome $a$, whereas the second party performed measurement   $y$ and got  outcome $b$  is denoted by $(a,b|x,y)$. The next step is to consider  physical theories,  that assign probabilities $p(a,b|x,y)$ to all possible  events of a Bell experiment.  The collection of all joint conditional probabilities $p(a,b|x,y)$ is called a  \emph{full behavior}. 

Clearly, for any physical theory,    behaviors satisfy nonnegativity and normalization conditions,~i.e., 
\begin{align*}
 p(a,b|x,y) &\geq 0,\  \forall a,b,x,y, \\
\sum_{a,b} p(a,b|x,y) &=1,\  \forall x,y.
\end{align*}
Furthermore,  for reasonable physical theory,   the  behaviors that can be generated between space-like separated parties have the property  that each parties local  marginal distibution does not depend on the other parties choice of measurement, i.e., 
\be\label{eq:ns}
\begin{aligned}
 \sum_b p(a,b|x,y) &= \sum_b p(a,b|x,y'),\ \forall a,x, y\ne y' \text{ and }\\
\sum_a p(a,b|x,y) &= \sum_a p(a,b|x',y),\ \forall b,y, x\ne x'.
\end{aligned}
\ee
A full behavior that satisfies all the linear constraints given in \eqref{eq:ns} is called {\em no-signaling}. Given a no-signaling behavior, we denote by $p_A(a|x)$ and $p_B(b|y)$ the local marginal distributions of Alice and Bob, respectively.

As we only consider dichotomic measurements,  
we can use an equivalent parametrization of  a no-signaling behavior  $p(a,b|x,y)$ in terms of average values. Explicitly, for  outcome labels  $a,b \in\{\pm 1\}$ we use  the affine bijection 
\be\label{covmap}
p(a,b|x,y) \mapsto (c_x, c_y, c_{xy})
\ee
where
$$c_x = \sum_{a\in \{\pm 1\}}ap_A(a|x) \quad 
c_y = \sum_{b\in \{\pm 1\}} bp_B(b|y),$$ 
$$c_{xy} = \sum_{a,b\in \{\pm 1\}}abp(a,b|x,y).$$
The image of the set  of full behaviors  under the map \eqref{covmap}  is called the set of  {\em full correlators}, and its  elements are denoted by  
$(c_x, c_y, c_{xy})$. Lastly, the  coordinate projection of the set of full correlators on the coordinates $c_{xy}$ is  the {\em correlator space} and its element  are called \emph{correlators}.\\

\noindent {\em Local behaviors and correlators. } 	A behavior $p(a,b|x,y)$ is called local deterministic if $p(a,b|x,y) = \delta_{a,f(x)} \delta_{b,g(y)}$, where $\delta$ is the Kronecker delta function and $f,g$ are functions from the input set to the output set.  Furthermore, a behaviour $p(a,b|x,y)$ is called {\em local}, if it can be written as a convex combination of local deterministic behaviours, i.e.,  $		p(a,b|x,y) = \sum_i \mu_i p_i,$ where $\mu\geq 0$, $\sum_i \mu_i = 1$ and each $p_i$ is local deterministic. 

Fixing  outcome labels  $a,b \in\{\pm 1\}$, the corresponding set of correlators is  the convex hull of all $n\times m$ matrices  $xy^\top$, where  $x\in \{\pm 1\}^{n\times 1} $ and $ y\in \{\pm1 \}^{m\times 1}$.  This  is known as the \emph{cut polytope} of the complete bipartite graph $K_{n,m}$ (in $\pm 1$ variables) and is of central importance in the field of combinatorial optimization~\cite{DL}.


\medskip 
\noindent {\em Quantum behaviors and correlators.}
According to the (Hilbert space) axioms of quantum mechanics, a full behavior $p(a,b|x,y)$ is quantum, if there exist $\rho $, a  quantum state  acting  on the  Hilbert space $\mathcal{H}_A\otimes \mathcal{H}_B$, and  $\{M^x_a\}$ and $\{M^y_b\}$, local measurements acting on $\mathcal{H}_A$ and $\mathcal{H}_B$,~i.e., 
\begin{align*}
\rho \succeq 0  &\,\text{ and }\, \tr(\rho)=1;\\
M^x_a \succeq 0 &\,\text{ and }\, \sum_a M^x_a = 1;\\
M^y_b \succeq 0 &\,\text{ and }\, \sum_b M^y_b = 1\,.
\end{align*}
such that  $p(a,b|x,y)=\tr(M^x_a\otimes M^y_b\rho)$.
Equivalently, a full correlator  $(c_x,c_y, c_{xy})$ is quantum, if there exist  a  quantum state $\rho$   acting  on the  Hilbert space $\mathcal{H}_A\otimes \mathcal{H}_B$, and $\pm 1$ observables $A_1,...,A_n$, $B_1,..,B_m$ acting on  $\mathcal{H}_A$ and $\mathcal{H}_B$ respectively, such that
\be\label{qcorrelation}
\begin{aligned}
c_x&=\tr((A_x\otimes I)\rho), \ x\in [n];\\
c_y & =\tr((I\otimes B_y)\rho), \ y\in [m];\\
 c_{xy} & = \tr(\rho A_x\otimes B_y),\  x\in [n], y\in [m]\,.
\end{aligned}
\ee

The set of quantum correlators, denoted by $\co(n,m)$, is the set of all vectors $c_{xy}$ where  $c_{xy} = \tr(\rho A_x\otimes B_y)$ for a quantum state $\rho$ and $\pm 1 $ observables $A_1,...,A_n$, $B_1,..,B_m$. It is evident that the difference with full correlators lies in the lack of local marginals $a_x,b_y$. Note that the set of quantum correlators is a compact and convex subset of the cube $[-1,1]^{n m}$.
Throughout this work,  we arrange the entries of a quantum correlator $c_{xy}\in\co(n,m)$ as an $n\times m$ matrix  $C$, which we call a quantum correlation matrix. We use the vector and matrix representations  interchangeably.

\section{\label{sec:cor-analytic}Link to  semidefinite programming}\label{sec:linktosdp}


A   semidefinite program (SDP) is a mathematical  optimization problem, where the objective is to  optimize a linear function over an affine slice of the cone of positive semidefinite  matrices.  A SDP in primal canonical form is given by 
\begin{equation}\tag{P}\label{eq:primalsdp}
p^*=\underset{X}{\text{sup}} \left\{ \la C,X\ra  :    X \succeq  0,\   \la A_i,X \ra=b_i \ (i\in [\ell]) \right\},
 \end{equation}
where $\la \cdot, \cdot \ra$  denotes  the trace inner product of matrices,  and  the  generalized inequality 
 $X\succeq0$ means that the matrix  $X$ is positive semidefinite, i.e., it has nonnegative eigenvalues.
SDPs are  an important   generalization of linear programming, which is obtained as a special case when all involved matrices are diagonal. SDPs  have significant modeling power, powerful duality theory,  and efficient algorithms for solving them.

The link between quantum correlators and SDPs originates in  the seminal work of Tsirelson  \cite[Theorem~2.1]{TS87}. Specifically, Tsirelson showed  that 
   a    matrix  $C=(c_{xy})\in  [-1,1]^{n\times m}$ is a quantum correlation matrix if and only if there exists a collection of real unit vectors $u_1,\ldots,u_n,v_1,\ldots,v_m$ such that 
 \begin{align}\label{veveve}
 c_{xy}=\la u_x,v_y\ra, \ \text{ for all }  x\in [n], y\in [m]. 
\end{align}
We note that $\la \cdot, \cdot \ra$ denotes the canonical  inner product of vectors. As a consequence of  Tsirelson's theorem 
it follows   that $C=(c_{xy})\in  [-1,1]^{n\times m}$ is  a quantum correlation matrix if and only if  the following SDP is feasible:
\be\label{primal:completion}
\begin{aligned}
\underset{X}{\max} & \quad 0 \\
 \text{s.t.} & \quad X_{xy}=c_{xy}, \ x\in [n], y\in [m],  \\
& \quad X_{ii}=1, \ i\in [n+m],\\
 & \quad X\in \mathcal{S}^{n+m}_+ .
\end{aligned}
\ee

To establish  this  equivalence note that  if $c_{xy}=\la u_x,v_y\ra, \forall x\in [n], y\in [m],$ where   $\|u_x\|=\|v_y\|=~1$, the Gram matrix $\gram(u_1,\ldots,u_n,v_1,\ldots,v_m)$ is feasible  for~\eqref{primal:completion}. Conversely, if $X\in \mathcal{S}^{n+m}_+$ is feasible for~\eqref{primal:completion},  by the spectral theorem for real symmetric matrices, there exist  unit vectors $u_1,\ldots,u_n,v_1,\ldots,v_m$ such that $X=\gram(u_1,\ldots,u_n,v_1,\ldots,v_m)$. Clearly, this implies that $c_{xy}=\la u_x,v_y\ra,\ \forall x\in [n], y\in [m]$, and thus $C=(c_{xy})$ is a quantum correlation matrix.



 Furthermore, a   geometric interpretation of Tsirelson's theorem 
 is that  the set of $n\times m$ quantum correlation matrices $\co(n,m)$  is the image of  the set of $(n+m)\times (n+ m)$ positive semidefinite   matrices with  diagonal entries equal to one, denoted by $\calE_{n+m}$,  under the~projection
 \be\label{eq:projection}
 \Pi: \mathcal{S}^{n+m} \rightarrow  \R^{n\times m},  \quad  \begin{pmatrix}A& C\\C^\top& B\end{pmatrix} \mapsto C,
\ee
  i.e., we have that  $\co(n,m)=\Pi(\calE_{n+m})$. 
  
  Any  matrix  in  $\Pi^{-1}(C)\cap \calE_{n+m}$  is called  a {\em PSD completion}  of $C$. 
  Thus, checking whether a   matrix  $C=(c_{xy})\in  [-1,1]^{n\times m}$ is a quantum correlation matrix reduces to checking that the partially specified  matrix 
   $ \left(\begin{smallmatrix}?& C\\C^\top& ?\end{smallmatrix} \right)$
 can be completed to a full  PSD matrix with diagonal entries equal to one. For 
 example, the CHSH correlator
$
C={1\over \sqrt{2}}\left(\begin{smallmatrix}
1 & 1\\1 & -1
 \end{smallmatrix}\right),
 $
 is quantum, as the  corresponding partial~matrix 
 \be\label{cervrgrt}
 \left( \begin{array}{cc|cc}
 1 & a & {1/ \sqrt{2}}& {1/ \sqrt{2}}\\
 a & 1 & {1/ \sqrt{2}} & -{1/ \sqrt{2}}\\
 \hline
 {1/ \sqrt{2}} & {1/ \sqrt{2}}& 1 & b\\
 {1/ \sqrt{2}}& -{1/ \sqrt{2}}& b & 1
 \end{array}\right),
 \ee 
 admits a PSD completion, obtained for  $a=b=0$. 

The problem of completing a partially specified matrix into a full PSD matrix is an important instance of semidefinite programming, referred    to as the {\em PSD matrix completion problem}, e.g. see  \cite{monique} and references~therein. 

One of the most fruitful  approaches for studying the PSD matrix completion problem has been the  use of graph theory. Specifically, let  $G=([n], E)$ be a simple  undirected graph, whose edges encode the positions of the known entries of the matrix.  
The {\em elliptope or coordinate shadow of a graph} $G$, denoted by $\calE(G)$,  is defined as the image of $\calE_{n}$ under the coordinate~projection 
 \be
 \Pi_G: \mathcal{S}^{n} \rightarrow  \R^{E},  \quad  A\mapsto (A_{ij})_{ij\in E}.
\ee
In other words,  any  vector  $a\in \calE(G)\subseteq \R^E$ corresponds to a $G$-partial matrix, that admits a completion to a full PSD matrix with diagonal entries equal to one.  

The   properties and structure of the elliptope of  a graph have been studied extensively, e.g.~see~\cite{dima,DL,wolk,shin}. 
By the definition of the elliptope if a graph, it is clear that $$\co(n,m)=\calE(K_{n,m}),$$ 
where  $K_{n,m}$ denotes the complete bipartite graph, where the bipartitions have  $n$ and $m$ vertices respectively. This link allows us to utilize   properties concerning  the structure of the elliptope of a graph $\calE(G)$, in our study of the structure  of the set of quantum correlators. 

\section{Analytic description of $\co(n,m)$}\label{sec:analytric}
A first property of $\calE(G)$ of  relevance  to this work is that  the elliptope of a graph  $G=([n],E)$ is subset of a nonlinear transform of the  {\em metric polytope}, denoted by  $\met(G)$,  which consists    of all vectors $x=(x_e)\in\R^E$ satisfying:
\begin{align} 
& 0\le x_e\le 1, \text{ for all } e\in E;\label{box}\\
&  \sum_{e\in F} x_e-\sum_{e\in C\setminus F}x_e \le |F|-1\label{cycle},
\end{align}
where $C$ is any cycle in $G$,  and  $F$  is any odd cardinality subset of $C$. Recall that a cycle in a graph is a sequence of vertices starting and ending at the same vertex, where each two consecutive vertices in the sequence are adjacent to each other. We refer to the inequalities of type \eqref{box} and \eqref{cycle}   as  {\em box inequalities} and    {\em cycle inequalities} respectively. 

The relation between $\calE(G)$ and $\met(G)$
is completely understood, e.g. see~\cite[Theorem 4.7]{L97}. Specifically, its is known that
for any graph $G$, we have the inclusion  $\calE(G)\subseteq \cos(\pi(\met(G)))$. Furthermore, equality holds if and only if the graph $G$ does not have the complete graph on four vertices, denoted by $K_4$,  as a minor. 

Two observations concerning this result 
are in order.
 First, note that  the cosine function   is applied componentwise, i.e., for a vector $x=(x_e)\in\R^E$ we define  $\cos(x)\in \R^E,$ where $\cos(x)_e=\cos x_e.$   Second,  a  graph $H$ is called a  minor  of a graph $G$,  if $H$ can be obtained from $G$  through a series of edge deletions, edge contractions and isolated node deletions.

Based on this 
and the fact that $\co(n,m)=\calE(K_{n,m}),$ we now derive    an analytic description
 for $\co(n,m)$, whenever $\min\{n,m\}\le 2$.   Indeed,  it is easy to check  that   the complete bipartite graph $K_{n,m}$ has the  complete graph on four vertices   $K_4$ as  minor if and only if $\min\{n,m\}> 2$. Thus, it follows that 
  $\co(n,m)=\cos\pi(\met(K_{n,m})),$ whenever    $\min\{n,m\}\le~2$. This gives:


\begin{thm}\label{thm:main}For $\min\{n,m\}\le~2$, 
we have that  $C=(c_{xy})\in \co(n,m)$ if and only if the following  linear  system is feasible
\be
\begin{aligned}
&0\le \theta_{xy}\le \pi, \ \forall x,y,\\
& 0\le \ta_{1j}+\ta_{2i}+\ta_{2j} -\ta_{1i}\le 2\pi ,\\
& 0\le \ta_{1i}+\ta_{2i}+\ta_{2j} -\ta_{1j}\le 2\pi, \\
& 0\le \ta_{1i}+\ta_{1j}+\ta_{2j} -\ta_{2i}\le 2\pi, \\
& 0\le \ta_{1i}+\ta_{1j}+\ta_{2i} -\ta_{2j}\le 2\pi,
\end{aligned} 
\ee
where   $3\le i<j\le n+2$ and  $\theta_{xy}=\arccos(c_{xy})$.

\end{thm}

As an example, in  the   case of $\co(2,2)$, the resulting  characterization  
is known as the Tsirelson-Landau-Masanes criterion, which has been rediscovered  numerous times, e.g.  see ~\cite{TS87},~\cite{Landau} and \cite{masanes}.  
To the best of our knowledge, the  description for   $\co(2,n)$ when  $n\ge 3$ is new. The proof of Theorem \ref{thm:main} is given in the~Appendix. 

\medskip 

%
%
%
%
%

\section{Extremal and exposed~correlators}
In this section we use   that $\co(n,m)$ is a projection of the elliptope $\calE_{n+m}$, to  study the convex geometry of $\co(n,m)$. We begin with its extreme points 
(or zero-dimensional faces) 
and continue with its  exposed points.

\subsection{\label{sec:extreme}Extremal correlators}

A matrix $C\in\co(n,m)$  is an extreme point of $\co (n,m)$  if the equality $C=\lambda C_1+(1-\lambda)C_2$, where $\lambda\in (0,1)$ and $C_1,C_2\in \co(n,m)$, implies~$C=C_1=C_2$.  

The set of  extreme points of $\co(n,m)$, denoted by $\ext (\co(n,m))$,  is  important for the following  reasons.   Firstly,  since the set of quantum correlators is compact (i.e. closed \& bounded) and  convex, by the Krein-Milman Theorem (e.g.~see \cite[Theorem 3.3]{bar}), $\co(n,m)$  is equal to the convex hull of its extreme points. 

Secondly,  Tsirelson  showed that any quantum realization of an extremal correlator  in  $\co(n,m)$ corresponds to a complex representation of an appropriate Clifford algebra \cite[Theorem 3.1]{TS87}. As a  consequence, depending on the parity of the rank of an extremal correlator $C$, it either has one or two ``non-equivalent'' (up to arbitrary unitaries) quantum representations. In modern language, Tsirelson's work specialized to  the case of even-rank extremal correlators is a  self-testing statement, a  connection we pursue further  Section~\ref{sec:extreme-selftest}.


In this  section we derive  an exact characterization for  extremality in $\co(n,m)$,   in terms of the PSD matrix  completion problem. This fact essentially follows from the work of Tsirelson \cite{TS87, TS93}, and was also  recently  noted in~\cite[Theorem 3.3]{cpsd}. 
 The main tool  in the proof  is a set of  necessary conditions for extremality  derived by Tsirelson, which we have collected in  Theorem \ref{thm:extconditions} in the Apendix.

\begin{thm}\label{extpointschar}
A correlator $C\in \co(n,m)$ is extremal  if and only if $C$ has a unique PSD completion $\hc\in\calE_{n+m}$, and  furthermore, 
\be\label{rankcondition}
{\rm rank}(\hc\circ \hc)=\binom{{\rm rank}(\hc)+1}{2},
\ee
where $\circ $  denotes the componentwise product of matrices.  
\end{thm} 
\begin{proof}
The forward implication is a consequence of  
 Theorem~\ref{thm:extconditions} $(iii)$, combined with  the following characterization of extreme points of the elliptope~\cite{LT94}:
\be\label{extelliptope}
E\in \ext(\calE_n) \Longleftrightarrow  \rank(E\circ E)=\binom{\rank(E)+1}{2}.
\ee

For the converse direction,  by assumption we have that $\Pi^{-1}(C)\cap \calE_{n+m}=\{\hc\}$, where the map $\Pi$ was defined in \eqref{eq:projection}. Furthermore, the rank assumption on $\hat{C}$ combined with \eqref{extelliptope} imply that $\hc \in\ext(\calE_{n+m})$. Since $C\in \ext(\co(n,m))$ if and only if $\Pi^{-1}(C)$ is a face of $\calE_{n+m}$, e.g. see \cite[Lemma 2.4]{NLV}, we conclude that $C\in \ext(\co(n,m))$ (as an extreme point is a face).
\end{proof}
Illustrating the usefulness of  Theorem~\ref{extpointschar}, we   now show the extremality   of various  quantum correlation matrices. 

The only other  technique available in the literature for showing  extremality of a quantum correlator is via   the notion of self-testing. 
Specifically, it was shown in ~\cite[Proposition C.1]{geometry} that a full correlator  $ (c_x,c_y,c_{xy})$ which is a self-test,  is also necessarily an extreme point of the set of {\em full} correlators. 
It is easy to verify that this argument remains vaild for correlators, i.e., 
if $C\in \co(n,m)$ is a self-test, it is also an  extreme point of $\co(n,m).$

\medskip

\noindent {\bf Example 1:} The CHSH correlator 
$
C={1\over \sqrt{2}}\left(\begin{smallmatrix}
1 & 1\\1 & -1
 \end{smallmatrix}\right),
 $
is well-known to be a self-test (e.g. see \cite[Theorem 4.1]{outlook}), and thus, it is an extreme point of $\co(2,2)$.  To recover this by Theorem \ref{extpointschar},  we first show that the partial matrix
 \be\label{cervrgrt}
 \left( \begin{smallmatrix}
 1 & a & {1/ \sqrt{2}}& {1/ \sqrt{2}}\\
 a & 1 & {1/ \sqrt{2}} & -{1/ \sqrt{2}}\\
 {1/ \sqrt{2}} & {1/ \sqrt{2}}& 1 & b\\
 {1/ \sqrt{2}}& -{1/ \sqrt{2}}& b & 1
 \end{smallmatrix}\right)
 \ee 
 admits a unique  PSD completion. Indeed, consider an arbitrary completion and let  $x_1,x_2,y_1,y_2$ be the vectors in a  Gram decomposition. Define 
 $z_+={x_1+x_2\over \sqrt{2}}$ and 
 $ z_-={x_1-x_2\over \sqrt{2}}.$
 Clearly, $\la z_+,z_-\ra=0$. Furthermore, $\la y_1,z_+\ra=1$ and $\la y_1,z_-\ra=0$. Since $\|y_1\|=1$ it follows that  $y_1={z_+\over \|z_+\|},$
 and in the same manner   we get that 
  $y_2={z_-\over \|z_-\|}.$
This implies that $b=\la y_1,y_2\ra=0$. Similarly, we get that $a=0$. 
Thus the unique PSD completion is 
 \be\label{chshunique}
 \hc= \left(\begin{smallmatrix}
 1 & 0 & 1/ \sqrt{2}&  1/ \sqrt{2}\\
 0 & 1 & 1/ \sqrt{2} & -1/ \sqrt{2}\\
 1/ \sqrt{2} &  1/ \sqrt{2}& 1 & 0\\
 1/ \sqrt{2}& - 1/ \sqrt{2}& 0 & 1
 \end{smallmatrix}\right).
 \ee 
 Lastly, as $\rank(\hc)=2$ and $\rank(\hc\circ \hc)=3,$ it follows by Theorem \ref{extpointschar} that $C\in \ext(\co(2,2))$. 
 
 \medskip 
 \noindent {\bf Example 2:} The Mayers-Yao correlator  \cite{MY}
\be\label{MYcorr}
C=\left(\begin{smallmatrix}
1 & 0 &1/\sqrt{2}\\ 
0 & 1& 1/\sqrt{2}\\ 
1/\sqrt{2}& 1/\sqrt{2}& 1\end{smallmatrix}\right),
\ee
is a self-test (e.g. see \cite[Theorem 4.2]{outlook}), and thus, it  is an extreme point of $\co(3,3)$.
To recover this by Theorem~\ref{extpointschar} we first  check that the corresponding partial matrix 
\be\label{sdvrthr}
\left(\begin{smallmatrix}
1 & a&  b &  1 & 0 &1/\sqrt{2}\\
a & 1  & c & 0 & 1& 1/\sqrt{2}\\
b& c& 1& 1/\sqrt{2}& 1/\sqrt{2}& 1\\
1 & 0 & 1/\sqrt{2} &  1& d& e\\
0 & 1 & 1/\sqrt{2}&  d& 1 & f\\
1/\sqrt{2}& 1/\sqrt{2}  & 1& e & f& 1  
\end{smallmatrix}\right),
\ee
admits  a unique PSD completion. To see this, consider an arbitrary PSD  completion  and let $x_1, x_2, x_3,  y_1, y_2, y_3$ be a  Gram decomposition. Since $\|x_1\|=\|y_1\|=1$ and $\la x_1,y_1\ra=1$, we  have that $x_1= y_1$. Similarly, we get that  $x_2=y_2$. These two conditions imply that
\be
\begin{aligned}
a&=\la x_1,x_2\ra=\la x_1,y_2\ra=0,\\
b& =\la x_1,x_3\ra=\la y_1,x_3\ra=1/\sqrt{2},\\
c & =\la x_2,x_3\ra=\la y_2,x_3\ra=1/\sqrt{2},\\
d&=\la y_1,y_2\ra=\la x_1,x_2\ra=0,\\
e & =\la y_1,y_3\ra=\la x_1,x_3\ra=1/\sqrt{2},\\
f& =\la y_2,y_3\ra=\la x_2,y_3\ra=1/\sqrt{2}.
\end{aligned}
\ee 
Summarizing,  the unique PSD completion of $C$ is
\be\label{MY}
\hc=\left(\begin{smallmatrix}
1 &0 &  1/\sqrt{2} &  1 & 0 &1/\sqrt{2}\\
0 & 1  & 1/\sqrt{2}& 0 & 1& 1/\sqrt{2}\\
1/\sqrt{2}& 1/\sqrt{2}& 1& 1/\sqrt{2}& 1/\sqrt{2}& 1\\
1 & 0 & 1/\sqrt{2} &  1& 0& 1/\sqrt{2}\\
0 & 1 & 1/\sqrt{2}&  0& 1 & 1/\sqrt{2}\\
1/\sqrt{2}& 1/\sqrt{2}  & 1& 1/\sqrt{2}& 1/\sqrt{2}& 1  
\end{smallmatrix}\right).
\ee
Lastly, since $\rank(\hc)=2$ and $\rank(\hc\circ \hc)=3$  it follows by Theorem \ref{extpointschar} that  $C\in\ext(\co(2,3))$.

\medskip

\noindent {\bf Example 3:} The quantum correlator
$$C={1\over 2}\begin{pmatrix}
1 & 1\\1 & -2
\end{pmatrix}.$$
is a self-test \cite{Scarani}, and thus, an extreme point of $\co(2,2)$. It can be easily checked that the corresponding partial matrix has a  unique PSD completion given by 
$$\hc=\left(\begin{smallmatrix}
1 & -1/2 & 1/2 & 1/2\\
-1/2 & 1 & 1/2 & -1\\
1/2 & 1/2 & 1 & -1/2\\
1/2 & -1 & -1/2 & 1
\end{smallmatrix}\right).$$
As $\rank(\hc)=2$  and $\rank(\hc\circ \hc)=3$, it follows that $C\in \ext(\co(2,2))$.

\medskip
\noindent {\em Extreme points of $\co(2,2)$.}
In this section we   give an explicit characterization of the extreme points of $\co(2,2)$, in terms of  the angle parametrization from  Theorem \ref{thm:main}. 


\begin{thm}\label{thm:analytic} Let  $C=(c_{xy})\in \co(2,2) $ and define  $\theta_{xy}=\arccos(c_{xy})\in [0,\pi]$ for all $x\in \{1,2\}, y\in \{3,4\}$. 
\begin{itemize}
\item[$(i)$] If $\rank(C)=1,$ then $C$ is extreme iff it is local deterministic, i.e., $C=xy^\top, $ for $x,y\in \{\pm 1\}^2$.
\item[$(ii)$] If $\rank(C)=2,$ then $C$ is extreme iff it saturates exactly one of the  inequalities 
$$0\le \sum_{xy\ne x'y'}\theta_{xy}-\theta_{x'y'}\le 2\pi, $$
and at most one of the  inequalities  
$$0\le \theta_{xy}\le \pi,$$
where $ x,x'\in \{1,2\},\  y,y'\in \{3,4\}.$ 

\end{itemize}
\end{thm} 
The case $\rank(C)=1$ is straightforward so we mainly focus  on   the case $\rank(C)=2$.  To  prove  extremality, we use the assumptions of the theorem  to prove the existence of a unique completion that satisfies   \eqref{rankcondition}. Extremality then follows by Theorem~\ref{extpointschar}. 

For the converse direction,  we  translate  the assumption  of extremality, namely unique completability and the rank condition \eqref{rankcondition}  to the level of the parameters $\ta_{xy}=\arccos(c_{xy})$. 
As it turns out, these assumptions   imply that  the unspecified parameters $\ta_{12}, \ta_{34}$  are uniquely determined in any completion. In turn, this shows  that  one cycle  inequality and at most one box inequality are tight. 
The details are given in the~Appendix.

\subsection{Verifying extremality computationally}
The  examples given in the previous section    illustrate the usefulness of the characterization of extremality given in Theorem \ref{extpointschar}. Nevertheless, it is not clear whether Theorem \ref{extpointschar} leads to an algorithm  for testing extremality, as a priori, it is not immediately  obvious  how to systematically check whether the corresponding completion problem has a unique solution. 
We  address this issue using   the rich duality theory enjoyed by SDPs, summarized  in Theorems \ref{sdpthm} and  \ref{thmdegeneracy} in the~Appendix. 


Back to the completion problem, given $C=(c_{xy})\in\co(n,m),$  its PSD completions coincide with the set of solutions of the  SDP feasibility  problem \eqref{primal:completion}.

Next, we  dualize the SDP  \eqref{primal:completion}. For this, we first  write   \eqref{primal:completion} in primal canonical form (recall \eqref{eq:primalsdp}), i.e., 
\be\label{primal:completioncanonical}
\begin{aligned}
\underset{X}{\max} & \quad 0 \\
 \text{s.t.} & \quad \la E_{xy}, X\ra=c_{xy}, \ x\in [n], y\in [m],  \\
& \quad \la E_{ii}, X \ra =1, \ i\in [n+m],\\
 & \quad X\in \mathcal{S}^{n+m}_+ ,
\end{aligned}
\ee
where $E_{xy}={1\over 2 }(e_xe_y^\top+e_ye_x^\top).$

The dual of the SDP \eqref{primal:completioncanonical}  is given by 
\be\label{dual:completion}
\begin{aligned}
\underset{\lambda,Z}{\inf} & \quad \sum_{i=1}^{n+m}\lambda_i+\sum_{x=1}^n\sum_{y=1}^m \lambda_{xy}c_{xy}\\
 \text{s.t.} & \quad  \sum_{i=1}^{n+m}\lambda_iE_{ii}+\sum_{x=1}^n\sum_{y=1}^m  \lambda_{xy}E_{xy}=Z\in \mathcal{S}^{n+m}_+.
 \end{aligned}
 \ee 
Note that the  SDP  \eqref{dual:completion} admits  a positive definite feasible solution, e.g. obtained  by  setting $\lambda_{xy}=0, \ \forall x,y$ and taking  $\lambda_i$ to be sufficiently large. 
 Furthermore, by weak duality for SDPs (cf. Theorem \ref{sdpthm}~$(i)$), we have that $0=p^*\le d^*$, i.e., $d^*>-\infty$. By strong duality for SDPs (cf. Theorem \ref{sdpthm} $(iv)$), these two properties  imply  that  the value of the dual SDP \eqref{dual:completion}    is equal to zero, i.e. $d^*=~0$. Furthermore, $d^*=0$ 
  it is clearly attained, e.g. take $\lambda_i=\lambda_{xy}=0$. Lastly, as $C\in\co(n,m)$ by assumption, the primal SDP   \eqref{primal:completion} is also attained. Thus, to show that  the SDP \eqref{primal:completion} has a unique solution, it suffices  to  exhibit a  nondegenerate optimal solution for \eqref{dual:completion}.
  
  Specializing the  definition of dual nondegeneracy for SDPs (recall  \eqref{eq:dnondeg})  to a dual feasible solution $(\lambda,Z)$ for the  SDP \eqref{dual:completion}, this is  
   equivalent to asking that $M=0$ is the only solution of the system:
  \be\label{lsystem}
  \begin{aligned}
  MZ& =0; \\
   M_{ii} & =0,\ 1\le  i \le m+n; \\
   M_{ij}& =0, \   \ 1\le i\le n, n \le j \le n+m.
   \end{aligned}
   \ee
 An important  observation  is that \eqref{lsystem} is  a  linear program in the entries of the symmetric matrix variable $M\in \mathcal{S}^{n+m}$, and thus, it is efficiently solvable. 
 
We are now ready to describe an algorithmic procedure for determining  extremality of a given  $C\in\co(n,m)$, based  on Theorem~\ref{extpointschar} and the notion of SDP nondegeneracy. 
For the convenience of the reader, the algorithm  is   summarized  in a flow chart in Fig.~\ref{fig:flow}.
 
{\em Step~1:} We solve the pair of primal-dual SDPs \eqref{primal:completioncanonical}  and~\eqref{dual:completion}, to get $X_{\mathrm{opt}}$ and $Z_{\mathrm{opt}}$, respectively. 

{\em Step~2:} We check whether  $Z_{\mathrm{opt}}$ is dual non-degenerate, i.e., we check whether $M=0$ is the only solution to the linear programming problem \eqref{lsystem} (where $Z=Z_{\mathrm{opt}}$.)


{\em Step~3a:}  If $Z_{\mathrm{opt}}$ is non-degenerate, then $X_{\mathrm{opt}}$ is the unique solution of the primal SDP \eqref{primal:completioncanonical} by Theorem \ref{thmdegeneracy}. Lastly,  we check whether $\rank(X_{\mathrm{opt}} \circ X_{\mathrm{opt}}) = \binom{\rank(X_{\mathrm{opt}})+1}{2}$; If this holds then $C$ is extreme, and if it fails,  $C$ is not extreme. 

{\em Step~3b:} If $Z_{\mathrm{opt}}$ is degenerate and 
\be\label{scomp}
\rank(X_{\mathrm{opt}}) + \rank(Z_{\mathrm{opt}}) = m+n,
\ee
 we conclude that $C$ is {\em not extreme}. Indeed, if $C$ was extreme, by Theorem~\ref{extpointschar},   $ X_{\mathrm{opt}}$ would be the unique solution of the SDP \eqref{primal:completioncanonical}.  As $X_{\mathrm{opt}},Z_{\mathrm{opt}}$ satisfy 
 \eqref{scomp}, 
   by Theorem \ref{thmdegeneracy} $(ii)$ the matrix $ Z_{\mathrm{opt}}$ would be  dual nondegenerate optimal solution, a contradiction. 

{\em Step~3c:} If $Z_{\mathrm{opt}}$ is degenerate and 
\be\label{scomp2}
\rank(X_{\mathrm{opt}}) + \rank(Z_{\mathrm{opt}}) < m+n,
\ee
our procedure is inconclusive.

Note that by weak duality for SDPs, we always have that $\rank(X_{\mathrm{opt}}) + \rank(Z_{\mathrm{opt}}) \le  m+n$. Thus, condition \eqref{scomp} fails if and only if condition 
\eqref{scomp2} holds.




We implemented this procedure on MATLAB\textregistered \hspace{2pt} using the YALMIP package and Mosek  as the solver  \cite{implementation}. Furthermore, we tested the performance of the procedure on randomly generated  extremal points of  $\co (2,2)$. Specifically, with \texttt{randExtremeCorr22.m}, we generate a  random point in $\mathrm{ExtCor}(2,2)$  by randomly picking three angles, $\theta_1, \theta_2, \theta_3 \in (0,\pi)$, and setting  $\phi=\theta_1+\theta_2+\theta_3$. If $\phi<\pi$ or $2\pi<\phi<3\pi$, we set the fourth angle $\theta_4=\phi$, otherwise, we discard this instance. Then, by Theorem~\ref{thm:analytic}, the   correlator corresponding to $ \cos(\theta_1,\theta_2,\theta_3,\theta_4)$ is extremal in $\co(2,2)$. We applied our prodecure, called \texttt{extremeCorr.m}, on 1000 extremal points generated by \texttt{randExtremeCorr22.m}. In all instances, our algorithm correctly detected  that the generated points are indeed extreme. 

\begin{figure}[h]
\begin{center}
\includegraphics[width=8cm]{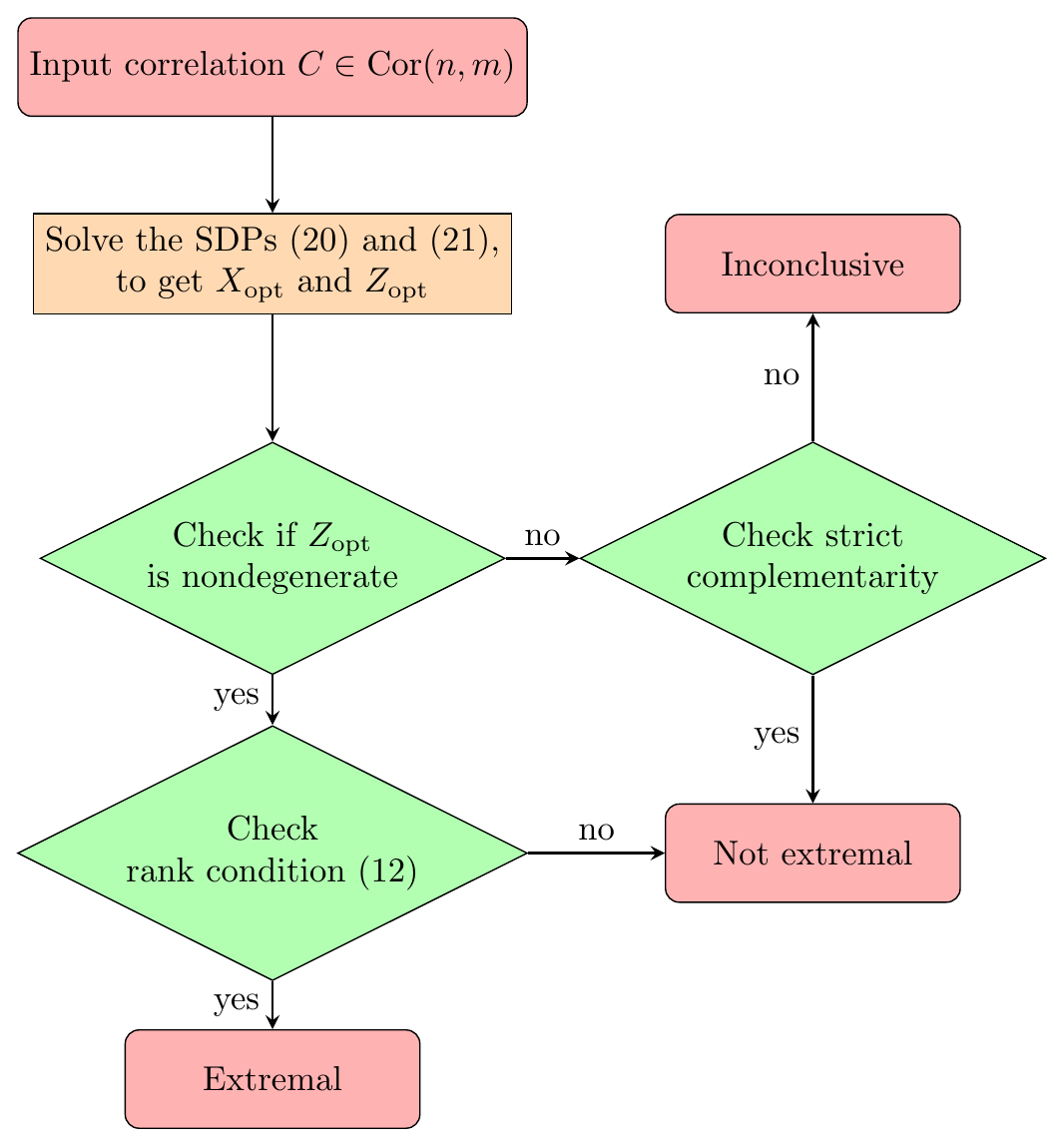} 
\caption{A flow chart describing the algorithmic procedure for determining extremality in $\co(n,m)$.}
\label{fig:flow}
\end{center}
\end{figure}


\subsection{\label{sec:extreme-selftest}Operational interpretation of extremality}
It turns out that the geometric concept of extremality has a nice operational interpretation. We now explain the connection with the task of self-testing, which has been hinted several times in the previous sections.

Self-testing, also referred to as device-independent characterization of the state and the measurements, or simply blind tomography, captures the idea  that certain  correlations between  space-like separated parties predicted by quantum theory determine the state and the measurement  up to local isometries  and other irrelevant degrees of freedom. 

 The term self-testing was  introduced   in the  work by Mayers and Yao~\cite{MY}. Nevertheless, the idea underlying self-testing was rediscovered  earlier numerous   times in the literature, for example in the works of Tsirelson \cite{TS87}, Summers-Werner \cite{SW}, and Popescu-Rohrlich \cite{PR}. The interested reader is  referred to \cite{outlook} for a general survey, and \cite{allpure,jed1} for  more recent developments. 
 
To  formally define self-testing, recall that a quantum correlator $C=(c_{xy})\in \co(n,m)$ can be written as 
\begin{align}\label{real}
c_{xy} = \psi^\dagger(A_x\otimes B_y)\psi, \,\,\text{ for all } x\in [n], y\in [m],
\end{align}
where  $\psi$ is a unit vector in $ \mathcal{H}_A\otimes \mathcal{H}_B$,  and the Hermitian operators   $A_x:~{\calH_A\to \calH_A}$, $B_y:~{\calH_B\to \calH_B}$   have eigenvalues  in~$[-1,1]$.  Any  ensemble $(\calH_A, \calH_B, \psi, \{A_x\}_x, \{B_y\}_y)$ with these properties that satisfies  \eqref{real} is called a  quantum realization of~$C$. 

Each  quantum realization $(\calH_A, \calH_B, \psi, \{A_x\}_x, \{B_y\}_y)$ of $C$  induces an entire orbit of realizations, obtained by applying local isometries and/or introducing irrelevant degrees of freedom. Concretely, given two   isometries $V_A:~\calH_A\to \calH_{A'}$ and $V_B: \calH_B\to \calH_{B'}$ (i.e., $V_A^\dagger V_A=1_{\calH_{A}})$ and $V_b^\dagger V_b=1_{\calH_{B}}$), we have that  the ensemble 
\be\label{ens:iso}
(\calH_{A'}, \calH_{B'}, (V_A\otimes V_B) \psi, \{V_AA_xV_A^\dagger\}_x, \{V_BB_yV_B^\dagger\}_y),
\ee is another quantum realization of $C$ since, 

$$\begin{aligned}
c_{xy}= &\psi^\dagger(A_x\otimes B_y)\psi=\\
=& \psi^\dagger (V_A^\dagger\otimes V_B^\dagger)(V_AA_xV_A^\dagger\otimes V_BB_yV_B^\dagger) (V_A\otimes V_B) \psi.
\end{aligned}
$$
Furthermore,   for any $\psi''\in \calH_{A''}\otimes \calH_{B''}$, the ensemble 
\be\label{ens:redundant}
(\calH_A\otimes \calH_{A''}, \calH_B\otimes \calH_{B''}, \psi\otimes \psi'', \{A_x\otimes 1_{{\calH}_{A''}}\}_x, \{B_y\otimes 1_{{\calH}_{B''}}\}_y),
\ee
is another quantum realization of $C$ since
$$ 
\begin{aligned}
c_{xy}= &\psi^\dagger(A_x\otimes B_y)\psi=\\
=& (\psi\otimes \psi'')^{\dagger}((A_x\otimes 1_{{\calH}_{A''}})\otimes (B_y\otimes 1_{{\calH}_{B''}}))(\psi\otimes \psi'').
\end{aligned}
$$
We say that the correlation $C$ self-tests the ensemble $(\calH_A, \calH_B, \psi, \{A_x\}_x, \{B_y\}_y)$ if all other quantum realizations of $C$ are of the form \eqref{ens:iso} or \eqref{ens:redundant} (or their combination), i.e.,  for any other quantum realization $(\calH_{A'}, \calH_{B'}, \psi, \{A'_x\}_x, \{B'_y\}_y)$ there exist   isometries $V_A:~\calH_{A}\otimes \calH_{A''}\to \calH_{A'}$ and $V_B: \calH_{B}\otimes \calH_{B''}\to  \calH_{B'}$ and a unit vector  $\psi''\in \calH_{A''}\otimes \calH_{B''}$ such that 
\be
\begin{aligned}
\psi'&=(V_A\otimes V_B)(\psi\otimes \psi''),\\
A_x' &=V_A(A_x\otimes 1_{\calH_{A''}})V_A^\dagger,\\
B_y' &=V_B(B_y\otimes 1_{\calH_{B''}})V_B^\dagger.
\end{aligned}
\ee
Furthermore, we say that $C$ is a self-test if it self-tests some   quantum realization   $(\calH_A, \calH_B, \psi, \{A_x\}_x, \{B_y\}_y)$.

In the following theorem we give a geometric characterization of self-testing for the special case of $\co(2,2)$. 
\begin{thm}\label{thm:extreme-selftest}Let $C\in \co(2,2)$ with $\rank(C)=2$. The following are equivalent:
\bi 
\item[$(i)$] $C$ is an extreme point of $\co(2,2)$.
\item[$(ii)$]  $C$ self-tests the singlet. 
\item[$(iii)$] $C$ is a self-test. 
\ei 
\end{thm}
\begin{proof}In  \cite[Theorem 1]{Scarani} it is shown that   a rank two correlator $C\in \co(2,2)$ self-tests the singlet if and only if, $C$ saturates exactly one of the inequalities
$$   -\pi \le \sum_{xy\ne x'y'}\arcsin(c_{xy})-\arcsin(c_{xy}) \le \pi, \ \forall x'y',$$
and at most one of the ineuqalities 
$$-\pi/2\le \arcsin(c_{xy})\le  \pi/2,\ \forall x,y,$$
where $ x,x'\in \{1,2\},\  y,y'\in \{3,4\}.$ 
Using that $\arccos(x)+\arcsin(x)={\pi\over 2}$ and $\ta_{xy}=\arccos (c_{xy}),$ the equivalence between   $(i)$ and $(ii)$ follows from Theorem~\ref{thm:analytic} $(ii)$.
Lastly, the equivalence between $(ii)$ and $(iii)$ is a special case of  \cite[Theorem 3.2]{TS87}. 
\end{proof}

\subsection{\label{sec:exposed}Exposed correlators}
An exposed face of $\co(n,m)$ is a subset $F\subseteq \co(n,m)$ for which there exists   a matrix $A\in\R^{n\times m}$ such that $F={\rm argmax}\{ \la A,X\ra: X\in \co(n,m)\}$.  A matrix $C\in \co(n,m)$ is an exposed point of $\co(n,m)$  if the singleton $\{C\}$ is an exposed face of $\co(n,m)$, i.e.,  there exists $A\in\R^{n\times m}$  such~that 
 $$\{C\}={\rm argmax}\left\{ \la A,X\ra: X\in \co(n,m)\right\}.$$

Setting $b=\max \{\la A,X\ra: X\in \co(n,m)\}$,  $C$ is an exposed point of $\co(n,m)$ if  the following two properties hold: $(i)$  $\la A,X\ra\le b, $ for all $X\in \co(n,m)$ and  $(ii)$ $ \la A, X\ra=b$ if and only if  $X=C$. In this setting, we say that  the hyperplane  $\mathcal{H}=\{X\in \R^{n\times m}: \la A,X\ra=b\}$ exposes  the point $C$.  

The exposed points of a convex set are always  extreme, but the converse is not always true. An example of such a point is the Hardy behavior \cite{hardy}, which is extreme point of the set of full behaviors (as it is a self-test \cite{hardyST}), but was recently shown to be non exposed~\cite{geometry}.

In this section, we use again SDP duality theory to give a sufficient condition for  a point $C\in \co(n,m)$ to be exposed. Our main tool is the following result. 

\begin{thm}\label{thm:exposed}
Let $C^*=(c^*_{xy})$ be an extreme point of  $\co(n,m)$  
and  $Z^*=
 \sum_{i=1}^{n+m}\lambda^*_iE_{ii}+\sum_{x=1}^n\sum_{y=1}^m \lambda^*_{xy}E_{xy}
 $
  a dual  optimal solution for \eqref{dual:completion}. 
  
\bi 
\item[$(i)$] The hyperplane
\be\label{hyperplane}
\mathcal{H} = \left\{(c_{xy}) \,:\, -\sum_{x=1}^n\sum_{y=1}^m\lambda_{xy}^*c_{xy} = \sum_{i=1}^{n+m} \lambda_i^*\right\},
\ee
supports the set $\co(n,m)$ at the point $C^*$, i.e., 
\begin{align*} 
& -\sum_{x=1}^n\sum_{y=1}^m\lambda_{xy}^*c_{xy} \le \sum_{i=1}^{n+m}\lambda_i^*, \ \forall C \in\co(n,m), \text{ and }\\
&-\sum_{x=1}^n\sum_{y=1}^m\lambda_{xy}^*c^*_{xy} = \sum_{i=1}^{n+m}\lambda_i^*.
\end{align*}
\item[$(ii)$] Furthermore, if  the homogeneous linear system 
 \be\label{system}
   MZ^* =0 \quad 
   M_{ii}  =0,\ 1\le  i \le n+m,
   \ee
   in the symmetric matrix variable $M\in \mathcal{S}^{n+m}$ has only the trivial solution $M=0$, then, the hyperplane $\mathcal{H}$ given in \eqref{hyperplane} exposes the point $C^*$. 
   \ei 
 \end{thm}
 
 \begin{proof}
 Recall that the solution set of \eqref{primal:completioncanonical} coincides with the set of PSD completions of $C^*$. As $C^*$ is extreme, by Theorem \ref{extpointschar}, it has a unique PSD completion $\hc^*\in \mathcal{S}^{n+m}_+$, i.e., $\hc^*$ is the unique solution of the SDP \eqref{primal:completioncanonical}.  
 

We have already seen   that the values of \eqref{primal:completioncanonical} and \eqref{dual:completion} coincide,  and  both are attained. Consequently, as $\hc^*$ and $Z^*$ are  primal-dual optimal, we have  $\la \hc^*, Z^*\ra=0$ (cf. Theorem \ref{sdpthm} $(iii)$). Expanding 
this we~get $$\sum_{i=1}^{n+m} \lambda_i^*+\sum_{x=1}^n\sum_{y=1}^m\lambda_{xy}^*c_{xy}^*=0.$$
Lastly, consider  an arbitrary $C\in \co(n,m)$ and let $\hc$ be one of its  PSD completions.  As the PSD cone is self-dual we get  that   $\la \hc ,Z_*\ra\ge 0 $,   and expanding, this gives 
$$\sum_{i=1}^{n+m} \lambda_i^*+\sum_{x=1}^n\sum_{y=1}^m\lambda_{xy}^*c_{xy}\ge 0,$$
which shows that $\mathcal{H}$ supports $\co(n,m)$ at $C^*$. 
Equivalently, $C^*$ is an optimal solution  of the~program:
\be\label{exposed}
\max\left\{ -\sum_{x=1}^n\sum_{y=1}^m\lambda_{xy}^*c_{xy} : \ C\in \co(n,m)\right\}.
\ee

Next, we further assume that  the linear system \eqref{system} admits only the trivial solution. To show that the hyperplane $\mathcal{H}$ exposes $C^*$, it suffices to show that  $C^*$ is the unique optimal solution of \eqref{exposed}. To do this, we first write  \eqref{exposed} as an SDP in  primal canonical form. 
Recalling that $\co(n,m)=\Pi(\calE_{n+m})$, it immediately follows  that \eqref{exposed} is equivalent to the SDP:
\be\label{SDPexposed}
\begin{aligned}
 \underset{X}{\max} \quad  & \la {\Lambda_b^*},X\ra  \\
 \text{ s.t.} \quad &  \hat{C}_{ii}=1, \ 1\le i\le n+m,\\
& \ X\in \mathcal{S}^{n+m}_+,
\end{aligned}
\ee
where  $\Lambda^*_{xy}=-\lambda_{xy}^*, \ \forall x\in [n], y\in [m]$,  and  
${\Lambda_b^*} =\left(\begin{smallmatrix} \Huge{0}_{n\times n} & {\Lambda^*\over 2}\\ {\Lambda^*\over 2}^\top & \Huge{0}_{m\times m}\end{smallmatrix}\right)$.
  The dual of \eqref{SDPexposed} is given by:
\be\label{SDPexposeddual}
\begin{aligned}
\underset{\lambda,Z}{\min} \ & \sum_{x=1}^n \lambda_x+\sum_{y=n+1}^{n+m} \mu_y \\
\text{ s.t.} \   &  \sum_{x=1}^n \lambda_xE_{xx}+\sum_{y=n+1}^{n+m} \mu_yE_{yy}-{\Lambda_b^*}=Z\in \mathcal{S}^{n+m}_+.
\end{aligned}
\ee
As the primal \eqref{SDPexposed} is strictly feasible and upper bounded, there exists  no duality gap and the dual is attained,  cf.  Theorem \ref{sdpthm} $(iv)$. To show that $C^*$ is exposed it remains to show that $C^*$ is the unique optimal solution of \eqref{SDPexposed}. For this, by   Theorem \ref{thmdegeneracy},  it suffices to show that the dual SDP \eqref{SDPexposeddual} has a nondegenerate optimal solution.

By the definitions  of $Z^*$ and ${\Lambda_b^*}$ we have that 
$Z^*=
 \sum_{i=1}^{n+m}\lambda^*_iE_{ii}-{\Lambda_b^*},$
 i.e., $Z^*$ is dual feasible for \eqref{SDPexposeddual}. Furthermore, as $\la \hc^*, Z^*\ra=0$, and $ \hc^*, Z^*$ are primal-dual feasible for \eqref{SDPexposed} and \eqref{SDPexposeddual} respectively, they are primal-dual optimal. Lastly, the assumption \eqref{system} implies that $Z^*$ is dual nondegenerate, and the proof is~concluded. 
  \end{proof}

Next, we illustrate the usefulness of Theorem~\ref{thm:exposed} by two concrete examples, followed by  a summary and the conclusions of our computational~work.

%
%
%
%
%

\medskip 
\noindent {\bf Example 4:} 
We show that the hyperplane 
$$c_{11}+c_{12}+c_{21}-c_{22}\le {2 \sqrt{2}},$$
exposes the CHSH correlator.
 We have already seen that the matrix $\hc$ given in  \eqref{chshunique} is the unique optimal solution for \eqref{primal:completioncanonical}. As $\rank(\hc)=2$, the nullspace of $\hc$ has dimension two, and a linear basis is given by 
 $$v_1=({1\over \sqrt{2}},{1\over \sqrt{2}},-1,0)^\top \quad v_2=({1\over \sqrt{2}},-{1\over \sqrt{2}},0,-1)^\top.$$
Using these two vectors we define
$$ Z^*=v_1v_1^\top+v_2v_2^\top=\left(\begin{smallmatrix}
 1 & 0 & -{1\over \sqrt{2}}& -{1\over \sqrt{2}}\\
 0&  1& -{1\over \sqrt{2}}& {1\over \sqrt{2}}\\
- {1\over \sqrt{2}} & -{1\over \sqrt{2}} &  1& 0\\
 -{1\over \sqrt{2}} & {1\over \sqrt{2}} & 0 &  1
 \end{smallmatrix}\right).
 $$
 Next, we show that  $Z^*$ is dual optimal for \eqref{dual:completion}. Indeed, by construction $Z^*$  is  feasible for \eqref{dual:completion}, and satisfies $\la \hc, Z^*\ra=\la \hc,v_1v_1^\top\ra+\la \hc,v_2v_2^\top\ra=0$. As $\hc$ is  optimal for \eqref{primal:completion}, Theorem \eqref{sdpthm} $(iii)$ implies that $Z^*$ is dual optimal. 
 
 Having established that $Z^*$ is dual optimal,  Theorem~\ref{thm:exposed} implies that  the~hyperplane 
 $c_{11}+c_{12}+c_{21}-c_{22}\le {2 \sqrt{2}},$
 supports  $\co(2,2)$  at the CHSH correlator. Lastly, to prove  that this hyperplane exposes 
 the CHSH correlator, by Theorem \ref{thm:exposed} $(ii)$, it suffices to show that the homogeneous linear system \eqref{system} only admits the trivial solution. A straightforward calculation reveals this is  the case.

 \medskip 
\noindent {\bf Example 5}. 
We show that the hyperplane
\be\label{myhyperplane}
\begin{aligned}
 -12\sqrt{2}c_{14}&+4c_{15}-4\sqrt{2}c_{16}+4c_{24}-12\sqrt{2}c_{25}-4\sqrt{2}c_{26}\\-4\sqrt{2}c_{34}
&-4\sqrt{2}c_{35}+2(2-3\sqrt{2})c_{36}\le 6(5\sqrt{2}+2),
\end{aligned}
\ee
 exposes the Mayers-Yao correlator    \eqref{MYcorr}.
 In Example 2, we showed  that   the SDP \eqref{primal:completioncanonical} has the unique solution 
\be
X^*=\left(\begin{smallmatrix}
1 &0 &  1/\sqrt{2} &  1 & 0 &1/\sqrt{2}\\
0 & 1  & 1/\sqrt{2}& 0 & 1& 1/\sqrt{2}\\
1/\sqrt{2}& 1/\sqrt{2}& 1& 1/\sqrt{2}& 1/\sqrt{2}& 1\\
1 & 0 & 1/\sqrt{2} &  1& 0& 1/\sqrt{2}\\
0 & 1 & 1/\sqrt{2}&  0& 1 & 1/\sqrt{2}\\
1/\sqrt{2}& 1/\sqrt{2}  & 1& 1/\sqrt{2}& 1/\sqrt{2}& 1  
\end{smallmatrix}\right).
\ee
Note  that  $\rank(X^*)=2$, and  in fact, its  column space is  spanned by the first two  columns.  Thus, its nullspace has dimension four, and a basis  is given by the vectors:
$$\begin{aligned}
v_1&=(-1,-1,-1,1,1,1)^\top, \\
v_2&=(-1,1,0,1,-1,0)^\top,\\
v_3&= (1,1,- \sqrt{2}, 1,1, - \sqrt{2})^\top,\\
v_4 &=(1,1,-1,-1,-1,1)^\top.
\end{aligned}$$
 Using these vectors we define,
\begin{align*}
 Z^*&=2\sqrt{2}v_1v_1^\top+(3\sqrt{2}+1)v_2v_2^\top+v_3v_3^\top+\sqrt{2}v_4v_4^\top\\
&=\left(\begin{smallmatrix}
 2(3\sqrt{2}+1) & 0 & 0 & -6\sqrt{2} & 2 & -2\sqrt{2}\\
 0 & 2(3\sqrt{2}+1) & 0& 2&-6\sqrt{2} & -2\sqrt{2}\\
 0 & 0& 3\sqrt{2}+2 & -2\sqrt{2} & -2\sqrt{2} & 2-3\sqrt{2}\\
 -6\sqrt{2} & 2 & -2\sqrt{2} & 2(3\sqrt{2}+1)& 0 & 0\\
 2 & -6\sqrt{2} & -2\sqrt{2} & 0 & 2(3\sqrt{2}+1)& 0\\
 -2\sqrt{2} & -2\sqrt{2} & 2-3\sqrt{2} & 0 & 0& 3\sqrt{2}+2
 \end{smallmatrix}\right).
 \end{align*}
 
 By construction, $Z^*$ is positive semidefinite,  feasible for \eqref{dual:completion},  and  satisfies  $\la X^*, Z^*\ra=~0.$ Consequently, by Theorem \ref{sdpthm} $(iii)$ we get that $Z^*$ is dual optimal for \eqref{dual:completion} and thus, by 
 Theorem~\ref{thm:exposed} $(i)$,  we see that \eqref{myhyperplane} is a valid hyperplane for $\co(3,3)$.  It remains to  show that the hyperplane \eqref{myhyperplane}  exposes the Mayers-Yao correlator. For this,  by Theorem~\ref{thm:exposed} $(ii)$, suffices to show that the linear system \eqref{system} only admits the trivial solution. An easy calculation shows that this is indeed the case.

 \medskip \noindent {\em Verifying exposedness computationally. } Theorem \ref{thm:exposed} leads to an algorithm  for checking whether  a given extremal correlator $C$  is exposed. This is summarized below:
  
{\em Step~1.} Solve the SDP   \eqref{dual:completion} to find an optimal solution $Z^*=
 \sum_{i=1}^{n+m}\lambda^*_iE_{ii}+\sum_{x=1}^n\sum_{y=1}^m \lambda^*_{xy}E_{xy}
 $. 
 
  {\em Step~2.} Solve the SDP \eqref{SDPexposeddual} to find an optimal solution $Z$. If $Z$ is non-degenerate then $C$ is exposed. If $Z$ is degenerate, the test is inconclusive. 

We  implemented this procedure  on 1000  randomly generated extremal correlators from $\co(2,2)$. In all instances, our  algorithm  concluded that the corresponding correlators were also exposed.  The algorithm for generating the random instances is implmented  \texttt{randExtremeCorr22.m} and the entire procedure is implemented in  \texttt{exposedCorr.m}.  Our computations suggest that for $\co(2,2)$, most  extreme points are also exposed~\cite{implementation}.

\section{Conclusions and future work}
In this paper we studied geometric features  of the set of quantum correlators using  semidefinite programming.
Our starting point is that the set of quantum correlations can be seen as the projection of the feasible region of a semidefinite program, known as the elliptope. This  connection leads to  a characterization of its boundary, which generalizes  the well-known Tsirelson-Landau-Masanes  criterion (Theorem~\ref{thm:main}). Furthermore, based on this connection, we  were able to translate  results  concerning  the geometry of elliptopes to the set of quantum correlations. The first  question we considered was to  characterize   its extreme points,  or equivalently its zero-dimensional faces. We managed to  give a  complete characterization by making a link to  the positive semidefinite  matrix completion problem (Theorem \ref{extpointschar}). 
Furthermore, for the simplest Bell scenario we determined an explicit characterization of its extreme points (Theorem~\ref{thm:analytic}). Next, we gave a sufficient condition for a correlator  to be  the exposed   
 (Theorem~\ref{thm:exposed}). Our numerical experiments  suggest  that most nonlocal extreme points of $\co(2,2)$ are in fact exposed. Lastly,  we show that in the simplest Bell scenario,   the geometric property of extremality coincides with the operational task of self-testing (Theorem~\ref{thm:extreme-selftest}).

Our investigations in this paper naturally lead to several future directions:
\begin{itemize}
\item Can one obtain further analytic characterizations for scenarios not captured by Theorem~\ref{thm:main}? Can one generalize to multipartite correlation scenarios? 
\item What is the facial structure of the set of quantum~correlations? 
\item In the set of full quantum  behaviors, is extremality still equivalent to self-testing? If not, is extremality equivalent to self-testing with global isometries?
\end{itemize}
The first two questions are evident; let us comment on the third one. Here {\em self-testing with global isometries} is a similar notion to self-testing, but with the ``gauge" equivalence being relaxed to arbitrary global isometries (yet still preserve the observed behavior). In other words, the (equivalence) orbit of each realization is larger as we allow  global isometries
 in addition to local isometries. Note that self-testing with global isometries implies the usual self-testing, but the converse does not hold in general. Now it turns out that Theorem~\ref{thm:extreme-selftest} can be strengthened by adding the equivalence:
\begin{itemize}
\item[$(iv)$]  $C$ is a self-testing with global isometries. 
\end{itemize}
Thus, extremality gives a stronger property than (usual) self-testing in this context.  We leave  the study of these concepts and their relationships as future work.

\acknowledgements
We would like to thank J. Kaniewski, V. Scarani, and K.T. Goh for helpful discussions. LPT is supported by by the Singapore Ministry of Education Academic Research Fund Tier 3 (Grant No. MOE2012-T3-1-009); by the National Research Fund and the Ministry of Education, Singapore, under the Research Centres of Excellence programme. YC is supported by the John Templeton Foundation Grant 60607 ``Manybox locality as a physical principle''. AV is supported by the NUS Young Investigator award   R-266-000-111-133 and by an NRF Fellowship (NRF-NRFF2018-01 \& R-263-000-D02-281).

\appendix

\section{Semidefinite programming} 
In this section we briefly collect all the tools from SDP duality theory that we use in this paper. For proofs of these facts and additional details, the interested reader is referred to \cite{etienne}.

\begin{thm}\label{sdpthm}
Consider a pair of primal-dual SDPs

\begin{align}
p^*&=\underset{X}{{\rm sup}} \left\{ \la C,X\ra  :    X \succeq  0,\   \la A_i,X \ra=b_i \ (i\in [\ell]) \right\},\tag{P}\label{primalapp}\\
d^* & =\underset{y,Z}{{\rm inf}}\left\{ \sum_{i=1}^\ell b_iy_i\  :\  \sum_{i=1}^\ell y_iA_i-C=Z \succeq 0\right\}\tag{D}\label{dualapp}.
\end{align}
The following properties hold:
\bi 
\item[$(i)$] (Weak duality) Let $X, (y,Z)$ be a pair of primal-dual feasible solutions for $(P)$ and $(D)$ respectively. Then, $\la C,X\ra\le  \sum_{i=1}^\ell b_iy_i$, i.e., $p^*\le d^*.$
\item[$(ii)$] (Optimality condition) Let  $X, (y,Z)$ be a pair of primal-dual feasible solutions for $(P)$ and $(D)$ respectively. If $\la C,X\ra= \sum_{i=1}^\ell b_iy_i$, then we have that  $p^*=d^*$ and furthermore, $X$ and $(y,Z)$ are primal-dual optimal solutions respectively. 

\item[$(iii)$] (Complementary slackness) Let  $X, (y,Z)$ be a pair of primal-dual feasible solutions for $(P)$ and $(D)$ respectively. Under the assumption that $p^*=d^*$ we have that $X, (y,Z)$ are primal-dual optimal if and only if $\la X, Z\ra=0.$
\item[$(iv)$]  (Strong duality) Assume that $ d^*>-\infty$ (resp. $p^* <+\infty$) and that (D) (resp. (P)) is strictly feasible. Then $p^*=d^*$ and furthermore, the primal (resp. dual) optimal value is attained. 
\ei
\end{thm}

Given a  pair of primal-dual SDPs
\eqref{primalapp} and \eqref{dualapp}, 
a  primal feasible solution  $X$  is  called  {\em primal nondegenerate} if 
\begin{equation}\label{eq:nondeg}
 \mathcal{T}_X+{\rm span}\{ A_1,\ldots,A_\ell\}^{\perp}=\mathcal{S}^n,
 \end{equation}
and a dual feasible solution   $(y,Z)$ is   {\em dual nondegenerate}~if 
\begin{equation}\label{eq:dnondeg}
\mathcal{T}_Z+{\rm span} \{ A_1,\ldots,A_\ell\} =\mathcal{S}^n,
\end{equation}
where $\mathcal{T}_Z$ is the tangent space on the manifold of symmetric $n\times n$ matrices with rank equal to $\rank(Z)$,  at the point~$Z$, and  the sum of two vectors spaces denotes  the linear span of their~union.  

A concrete expression for the tangent space is   
$$\mathcal{T}_Z^\perp=\{M\in \mathcal{S}^n: MZ=~0\},$$ e.g. see \cite{complementarity} or  \cite[Lemma 7.1.1]{lov}.  

The next result summarizes sufficient conditions for the unicity of optimal solutions to SDPs identified in  \cite{complementarity}, which we use extensively throughout this work.

\begin{thm}\label{thmdegeneracy}
Consider a pair of primal-dual  SDPs
\eqref{primalapp} and \eqref{dualapp}, where we 
assume that their  optimal values are equal and that both are attained. We have that:
\bi
\item[$(i)$]  If~\eqref{primalapp} has a   nondegenerate optimal solution, \eqref{dualapp} has a   unique optimal solution.
Symmetrically, if~\eqref{dualapp} has  a   nondegenerate optimal solution, then~\eqref{primalapp} has a unique optimal  solution.

\item[$(ii)$] Furthermore, let $X, (y,Z)$ be a pair of primal-dual  optimal solutions that satisfy  
$$\rank(X)+\rank(Z)=n,$$
a property known as strict complementarity. Then, if  $X$ is the unique optimal solution for \eqref{primalapp},  $(y,Z)$ is dual nondegenerate. Symmetrically, if $(y,Z)$ is the unique optimal  for \eqref{dualapp},  $X$ is primal nondegenerate. 
\ei 
\end{thm}

\section{Necessary conditions for extremality in the set of quantum correlators}
In this section we collect several useful properties of extreme points of $\co(n,m)$, identified in the seminal work of Tsirelson  \cite{TS87, TS93}. For a more modern proof of these facts the reader is referred to \cite{prakash}.

A family   of vectors $u_1,...,u_n,v_1,...,v_m$ is  called a {\em $C$-system}  of $C\in \co(n,m)$ if they satisfy  $\norm{u_x}\leq~1$, $\norm{v_y}\leq~1$, and $c_{xy}=\la u_x,v_y\ra,~\forall x\in[n], y\in[m]$. 

\begin{thm}\label{thm:extconditions}
 For any  $C\in \ext(\co(n,m))$ we have:
\begin{itemize}
\item[$(i)$] All $C$-systems   are necessarily unit vectors; 
\item[$(ii)$] For any  $C$-system $\{u_1,\ldots u_n,v_1,\ldots, v_m\}$ we have  that   ${\rm span}(\{u_i\}_{i=1}^n)={\rm span}(\{v_j\}_{j=1}^m)$;
\item[$(iii)$]  $C$ admits a unique PSD completion, i.e., there exists a {\em unique} matrix  $\hc\in\calE_{n+m}$ with  
$\hc=\left(\begin{smallmatrix}A& C\\C^\top& B\end{smallmatrix}\right)\in \calE_{n+m}$.
Furthermore,  we have  that  $\hc
\in\ext(\calE_{n+m})$ and     ${\rm rank}(\hc)={\rm rank}(A)={\rm rank}(B)={\rm rank}(C)$.
\end{itemize}
\end{thm}
 We note that the proof of Theorem \ref{thm:extconditions}  establishes  the following chain of implications $(i) \Longrightarrow (ii)\Longrightarrow (iii) $. To the best of our knowledge, it is not known whether any of these three conditions is equivalent to extremality.

 \section{Properties of PSD  completions} 
 
 In this section we collect certain properties of  PSD completions,  which we use in the proof of Theorem \ref{thm:analytic}. Recall that any vector $x=(x_e)\in \calE(G)$ corresponds to a partial $G$-matrix that admits a PSD completion to a full PSD matrix with diagonal entries equal to one. 

 As any principal submatrix of a PSD matrix is also PSD, a necessary condition for $a\in \calE(G)$ is that the restriction of  $a$ to any completely specified principal submatrix is PSD. In other words, if $K$ is a clique in $G$, i.e., a fully connected subgraph of $G$, the restriction of $x$ to $K$, denoted by $x_K$, should lie in $\calE(K)$.   
 
 This necessary condition turns out to be sufficient if and only if the graph $G$ is chordal, i.e.,  every circuit of length at least four in $G$ has a chord. 
  \begin{thm}\cite{wolk}\label{wolk}
Graph $G$ is chordal iff $$\calE(G)=\{x\in \R^E: x_{K}\in\calE(K) \text{ for each clique } K\subseteq~G\}.$$
\end{thm}

We also need the following result, which gives an explicit description of $\calE(K_3)$. 

\begin{thm}\label{33case}\cite{BJT} 
Let $0\le \ta_1, \ta_2, \ta_3\le \pi$. Then, the matrix 
$$
C=\begin{pmatrix}1 & \cos \ta_1& \cos \ta_3\\
\cos \ta_1& 1 & \cos \ta_2\\
\cos \ta_3& \cos \ta_2 & 1 
\end{pmatrix}$$
is positive semidefinite if and only if 
\begin{align}
\ta_1\le \ta_2+\ta_3, & \quad \ta_2\le \ta_1+\ta_3, \\ 
\quad \ta_3\le \ta_1+\ta_2, & \quad \ta_1+ \ta_2+\ta_3\le 2\pi.
\end{align}
Furthermore, $C$ is singular if and only  if one of the above inequalities holds with equality. 
\end{thm}

 Based on the previous two results, we now prove the following:
 
\begin{lmm}\label{cor:uniqueness}Let $C\in [-1,1]^{2\times 2}$ and set $\theta_{xy}=\cos^{-1}(c_{xy})\in [0,\pi]$, for  $x\in \{1,2\},y\in \{3,4\}$. Then, $C$ has a PSD completion  if and only if: 
\begin{equation}\label{interval}
\begin{aligned}
&\max\{|\ta_{31}-\ta_{41}|,
|\ta_{32}-\ta_{42}|\} \le \\
&\le  \min\{\ta_{31}+\ta_{41},
\ta_{32}+\ta_{42},
2\pi -(\ta_{31}+\ta_{41}),
2\pi- (\ta_{32}+\ta_{42})\}.
\end{aligned}
\end{equation}
Furthermore, assuming  that \eqref{interval} holds,  for  any PSD completion of $C$, we have that  $\ta_{34}=\arccos(c_{34})$ lies in the interval \eqref{interval}. Conversely, for any $\ta$ in the interval \eqref{interval},  there exists a PSD completion of $C$ with $\ta_{34}=\ta$.
\end{lmm} 
 
 \begin{proof}
 Let $K_{2,2}$ be the graph with vertex set $\{1,2,3,4\}$ and edges $\{(1,3), (1,4), (2,3), (2,4)\}$. By  definition of the elliptope of a graph,  $C\in \calE(K_{2,2})$  if and only if  there exists $c_{34}$ such that 
$(C,c_{34})\in \calE(K_{2,2}\cup~(3,4)\})$. Nevertheless,  the graph $K_{2,2}\cup~\{3,4\}$ is chordal, and thus, by Theorem \ref{wolk} we have that   $C\in \calE(K_{2,2})$ if and only if there exists $c_{34}\in [-1, 1]$ such that $(c_{13}, c_{14}, c_{34})\in \calE(K_3) $ and  $(c_{23}, c_{24}, c_{34})\in \calE(K_3)$. Lastly, by  
Theorem \ref{33case},
these two conditions are equivalent to
the existence of  $\ta_{34}\in [0, \pi]$ satisfying the following sixteen inequalities:
\be\label{conditions0}
\begin{aligned}
\ta_{13}\le \ta_{14}+\ta_{34}, & \quad  \ta_{23}\le \ta_{24}+\ta_{34}\\
 \ta_{14}\le \ta_{13}+\ta_{34},  & \quad \ta_{24}\le \ta_{23}+\ta_{34} \\
  \ta_{34}\le \ta_{13}+\ta_{14}, & \quad \ta_{34}\le \ta_{23}+\ta_{24}\\
   \ta_{13}+\ta_{14}+\ta_{34}\le 2\pi,&   \quad \ta_{23}+ \ta_{24}+\ta_{34}\le 2\pi.
\end{aligned}
\ee
Eliminating $\ta_{34}$ from the  system \eqref{conditions0} we get  \eqref{interval}.
\end{proof}

\section{Proof of Theorem \ref{thm:main}} 


  Let $K_{2,n}$ be the complete bipartite graph, where the first  bipartition has two vertices labelled  $\{1,2\}$ and the second bipartition has $n$ vertices labelled  $\{3,\ldots,n+2\}$.  As $K_{2,n}$ has no $K_4$ minor, by \cite[Theorem 4.7]{L97} we have that $\co(2,n)=\cos\pi(\met(K_{2,n}))$. Setting   $\theta_{xy}=\arccos(c_{xy})$,  we get that 
$c=(c_{xy})\in \co(2,n)$ if and only if  there exists  $a=(a_{xy})\in \met(K_{2,n}) $ such that   $c_{xy}=\cos(\pi a_{xy}),$~i.e.,
$${\ta_{xy}\over \pi}\in \met(K_{2,n}).$$
The box constraints  for $\met(K_{2,2})$~give
$$0\le \theta_{xy}\le \pi, \ \forall x,y.$$

We continue with the cycle inequalities of $\met(K_{2,n})$. Note that for each $3\le i<j\le n+2$, the graph  $K_{2,n}$ contains    one cycle of length four,  namely $C=(1,i,j,2)$. The cycle inequality for $F=\{1i\}$ gives 
$$\ta_{1i}-\ta_{1j}-\ta_{2i}-\ta_{2j}\le 0,$$
and the cycle inequality for $C\setminus F$ gives
$$\ta_{1j}+\ta_{2i}+\ta_{2j}-\ta_{1i}\le 2\pi.$$
Summarizing, $c=(c_{xy})\in \co(2,n)$ if and only if 
\be
\begin{aligned}
&0\le \theta_{xy}\le \pi, \ \forall x,y,\\
& 0\le \ta_{1j}+\ta_{2i}+\ta_{2j} -\ta_{1i}\le 2\pi ,\\
& 0\le \ta_{1i}+\ta_{2i}+\ta_{2j} -\ta_{1j}\le 2\pi, \\
& 0\le \ta_{1i}+\ta_{1j}+\ta_{2j} -\ta_{2i}\le 2\pi, \\
& 0\le \ta_{1i}+\ta_{1j}+\ta_{2i} -\ta_{2j}\le 2\pi,
\end{aligned} 
\ee 
where   $3\le i<j\le n+2$.

\section{Proof of Theorem \ref{thm:analytic}. }

\noindent {\bf Part $(i)$. } Let $C\in \co(n,m)$ with ${\rm rank}(C)=1$. We show that $C$ is an  extreme point  if and only if   $C=xy^\top$, for some  $x\in \{\pm 1\}^n, y\in \{\pm 1\}^m$.

First, assume that $C=xy^\top$, where  $x~\in \{\pm 1\}^n, y\in~\{\pm 1\}^m$,  and consider a convex combination 
\be\label{wefegtr}
C =\sum_k \lambda_k C^k, \ \text{ where } \sum_k\lambda_k=1,\  \lambda_k\ge0,
\ee
where the matrices  $C^k$ lie in $\co(n,m)$, i.e., $C^k_{ij}=\la u^k_i,v^k_j\ra$, where $\|u^k_i\|=\|v^k_j\|=1$. 
 Note that   
\begin{align}
\label{sdvewe}
1=C_{ij}=|x_iy_j|&=\left|\sum_k \lambda_k C_{ij}^k\right|=\left|\sum_k \lambda_k \la u^k_i,v^k_j\ra\right| \\
&\le \sum_k \lambda_k  \left| \la u^k_i,v^k_j\ra\right|\le \sum_k \lambda_k=1,
\end{align}
and thus we have equality throughout. In particular, we get that 
$\sum_k \lambda_k  \left|
\la u^k_i,v^k_j\ra\right|=1$, 
and as $|\la u^k_i,v^k_j\ra|\le 1$, this implies  that 
$|\la u^k_i,v^k_j\ra|=1,$ for all $   k,i,j.$
In other words, all  matrices $C^k$ have entries $\pm 1$. Lastly, 
by \eqref{wefegtr} we get  that $C^k=xy^\top$ for all $k$, and thus   $C$ is~extremal. 

Conversely, let  $C$ is a rank-one extreme point of $\co(n,m)$. In this setting, we have already mentioned that  $C$ admits a unique PSD completion  $\hc\in\calE_{n+m}$ with  
$\hc=\left(\begin{smallmatrix}A& C\\C^\top& B\end{smallmatrix}\right)\in \calE_{n+m}$, and furthermore,   $\hc
\in\ext(\calE_{n+m})$ and     ${\rm rank}(\hc)={\rm rank}(A)={\rm rank}(B)={\rm rank}(C)$, e.g., see \cite[Lemma 2.5]{prakash}. By the assumptions we have that $\rank(C)=1$,  and thus $\rank(\hc)=1$, i.e., 
$$\hat{C}=\begin{pmatrix}x \\ y\end{pmatrix} \begin{pmatrix}x \\ y\end{pmatrix}^\top\in 
\calE_{n+m}.$$
Since $\hat{C}\in \calE_{n+m}$, it follows that $x_i^2=y_i^2=1$. In turn, this shows that   $C=xy^\top$ where $x\in \{\pm 1\}^n, y\in \{\pm 1\}^m$.

\medskip 

\medskip 
\noindent  {\bf Part $(ii)$. }  Let $C=\left(\begin{smallmatrix} c_{13} & c_{14} \\ c_{23} & c_{24}\end{smallmatrix}\right)\in\co(2,2)$ with $\rank(C)=2$, and set  $\theta_{xy}=\arccos(c_{xy})\in [0,\pi]$ for all $x,y$. We show that  $C\in \ext \co(2,2)$ if and only if it saturates exactly one of the  cycle inequalities 
$$0\le \sum_{xy\ne x'y'}\theta_{xy}-\theta_{x'y'}\le 2\pi, \ x,x'\in \{1,2\},\  y,y'\in \{3,4\}\,,$$ 
and at most one of the  box inequalities  
$$0\le \theta_{xy}\le \pi, \ x=1,2, \ y=3,4.$$

First, assuming that $C$ saturates exactly  one cycle inequality and  at most one box inequality, we show that $C$ is extremal. To show extremality, by Theorem \ref{extpointschar}, it suffices to show that $C$   has a unique PSD  completion  $\hc$,  that furthermore satisfies ${\rm rank}(\hc\circ \hc)=\binom{{\rm rank}(\hc)+1}{2}.$   For this, given an arbitrary PSD completion of the matrix~$C$: 
\be\label{partial}
\left(\begin{array}{cc|cc}1 & c_{12} & c_{13} & c_{14}\\
c_{12} & 1 & c_{23} & c_{24}\\
\hline
c_{13} & c_{23}& 1 & c_{34}\\
c_{14}& c_{24}& c_{34} & 1
\end{array}\right),
\ee
we show that $c_{12}$ and $c_{34}$ are uniquely determined. 
For concreteness, assume that the tight cycle inequality is:
\be\label{cvrghy}
\theta_{13}+\theta_{23}+\theta_{24}-\theta_{14}=0.
\ee
This assumption  is without loss of generality, as all cycle inequalities are equivalent up to permuting the parties and relabelling  the outcomes. In the optimization community this is known as  the ``switching symmetry'' of the cut polytope  \cite{DL}.

The fact that this equality leads to a unique completion should be evident by Lemma~\ref{cor:uniqueness}. However, let us be even more explicit and show that the unknown entries $c_{12},c_{34}$ are completely determined by this equation. Summing the  two triangle inequalities 
\be\label{sumtriangle}
-\theta_{13}-\theta_{23}+\theta_{12}\le 0 \quad  \text{ and } \quad  -\theta_{24}-\theta_{12}+\theta_{14}\le 0,
\ee
we get that 
\be\label{aggregate}
\theta_{13}+\theta_{23} +\theta_{24}-\theta_{14}\ge 0,
\ee
which combined with \eqref{cvrghy} implies that  
\be\label{entry1}
\theta_{12}=\theta_{13}+\theta_{23}=\theta_{14}-\theta_{24}.
\ee
Indeed, if either of the triangle inequalities in \eqref{sumtriangle} were strict, then \eqref{aggregate} would also be a strict inequality, contradicting \eqref{cvrghy}. 
Similarly, using the two triangle inequalities
\be
-\ta_{23}-\ta_{24}+\ta_{34}\le 0 \quad \text{ and } -\ta_{34}-\ta_{13}+\ta_{14}\le 0
\ee
we get that 
\be\label{entry2}
\ta_{34}=\ta_{23}+\ta_{24}=\ta_{14}-\ta_{13}.
\ee
Taking cosines in  \eqref{entry1} and  \eqref{entry2}, we see that  the two unspecified  entries  $c_{12}$ and $c_{34}$  in \eqref{partial} are uniquely determined. Specifically, we have: 
\begin{align} 
c_{12}&=\cos(\theta_{12})=c_{13}c_{23}-\sqrt{(1-c^2_{13})(1-c_{23}^2)};\\
c_{34}&=\cos(\theta_{34})=c_{23}c_{24}-\sqrt{(1-c^2_{23})(1-c_{24}^2)},
\end{align}
where we used that $\ta_{13}, \ta_{14}, \ta_{23}, \ta_{24}\in [0,\pi].$ 
Summarizing,   $C$  has  a PSD unique completion, denoted by $\hat{C}$.

The last step of the proof is to show that ${\rm rank}(\hc\circ \hc)=\binom{{\rm rank}(\hc)+1}{2}.$ For this,  let $x_1,\ldots, x_4$ be a Gram decomposition of $\hat{C}$. By \eqref{cvrghy}, these four  vectors span either a 1-dimensional or a 2-dimensional linear space. In particular, by projecting onto their linear span we may  assume that they  in fact lie  in $\R^2$. Thus, for the rank of $\hat{C}$ there are two cases to consider: $\rank(\hat{C})\in \{1,2\}.$

If $\rank(\hat{C})=1$, since $\rank(C)\le \rank(\hat{C})$ we have that $\rank(C)=1$,  contradicting the assumption that $\rank(C)=2$. Thus  we have that $\rank(\hat{C})=\rank(C)=~2.$ By Theorem \ref{extpointschar}, it remains to show that 
$$\rank(\hat{C}\circ \hat{C})=\dim {\rm span}(x_1x_1^\top, x_2x_2^\top, x_3x_3^\top,x_4x_4^\top)=3.$$

 Note that  $x_1$ is not parallel to $x_2$, for otherwise, the second row of $C$ would be a multiple of the first one, contradicting the assumption  $\rank(C)=2$. Similarly,  $x_3$ is not parallel to $x_4$.  
Thus, the sets   $\{x_1,x_2\}$ and $\{x_3,x_4\}$ are  basis for $\R^2$. Furthermore,   a simple calculation shows~that 
\be\label{x3}
x_3={\sin(\ta_{13})x_1+\sin(\ta_{23})x_2\over \sin(\ta_{13}+\ta_{23})},
\ee
and 
\be\label{x4}
   x_4={\sin(\ta_{14})x_1+\sin(\ta_{24})x_2\over \sin(\ta_{14}+\ta_{24})}.
   \ee
   For example, to see \eqref{x3}, expand $x_3$ in the $\{x_1,x_2\}$ basis, i.e., $x_3=\lambda x_1+\mu x_2$. Taking inner products with $x_1$ and $x_2$, and eliminating $\mu$ in the resulting linear system, we get  that $\lambda=(\cos(\ta_{13})-\cos(\ta_{23})\cos(\ta_{12}))/\sin^2(\ta_{12}).$ Lastly, substituting $\ta_{12}=\ta_{13}+\ta_{23}$, it follows that $\lambda=\sin(\ta_{13})/\sin(\ta_{13}+\ta_{23})$.
  Combining \eqref{x3} and \eqref{x4} we~get: 
\be\label{x3x3}
{\scriptstyle x_3x_3^\top={\sin^2(\ta_{13})x_1x_1^\top+\sin^2(\ta_{23})x_2x_2^\top+ 2\sin(\ta_{13})\sin(\ta_{23})(x_1x_2^\top+x_2x_1^\top)\over \sin^2(\ta_{13}+\ta_{23})}}
\ee
\be\label{x4x4}
{\scriptstyle x_4x_4^\top={\sin^2(\ta_{14})x_1x_1^\top+\sin^2(\ta_{24})x_2x_2^\top+ 2\sin(\ta_{14})\sin(\ta_{24})(x_1x_2^\top+x_2x_1^\top)\over \sin^2(\ta_{14}+\ta_{24})}}.
\ee
Next we show that 
\be\label{cefver}
x_1x_2^\top+x_2x_1^\top\not \in {\rm span}( x_1x_1^\top, x_2x_2^\top).
\ee
 Wlog we may assume that $x_1=(1,0)^\top$. This is because, for any unitary operator $U$, the vectors  $Ux_1,Ux_2, Ux_3,Ux_4$ also define  a Gram decomposition of $\hc$. Furthermore,  as $x_1=(1,0)^\top$, we have that $x_2=(\cos(\ta_{12}), \sin(\ta_{12}))^\top,$ and consequently, 
\begin{align}
x_1x_2^\top+x_2x_1^\top &=\begin{pmatrix} 2\cos(\ta_{12}) & \sin (\ta_{12})\\  \sin(\ta_{12})& 0 \end{pmatrix}\,,\\
x_1x_1^\top &=\begin{pmatrix} 1 & 0 \\0 & 0\end{pmatrix}\,,\\
x_2x_2^\top &=\begin{pmatrix}\cos^2(\ta_{12}) & \cos(\ta_{12})\sin(\ta_{12})\\\cos(\ta_{12})\sin\ta_{12}& \sin^2(\ta_{12})\end{pmatrix}\,.
\end{align}
As $\sin^2(\ta_{12})\ne0$ (since $x_1\not  \|x_2$) we see that \eqref{cefver} holds. 
\medskip 

Lastly, note that either $\sin(\ta_{13})\sin(\ta_{23})\ne0$ or $\sin(\ta_{14})\sin(\ta_{24})\ne0$. Indeed, if both are zero,  we would have  two tight box constraints, contradicting the hypothesis.  Wlog say  that $\sin(\ta_{13})\sin(\ta_{23})\ne0$. Since $x_1x_2^\top+x_2x_1^\top\not \in {\rm span}(x_1x_1^\top, x_2x_2^\top)$, it follows by \eqref{x3x3}~that 
$$x_3x_3^\top\not \in {\rm span}(x_1x_1^\top, x_2x_2^\top),$$
and thus $\dim {\rm span}(x_1x_1^\top, x_2x_2^\top, x_3x_3^\top,x_4x_4^\top)\ge 3.$ On the other hand, by \eqref{x3x3} and \eqref{x4x4} it follows that $\dim {\rm span}(x_1x_1^\top, x_2x_2^\top, x_3x_3^\top,x_4x_4^\top)\le  3.$
\medskip 

We now prove the converse  direction of the theorem. Say that $C$ is  a rank two  extreme point of $\co(2,2)$.   By Lemma \cite[Lemma 2.5]{prakash},  $C$ has a unique PSD completion  $\hc$, where $\rank(\hc)=\rank(C)=2$ and  $\rank(\hc\circ \hc)=3.$

First, note that  under the assumptions of the theorem there can be at most one tight box constraint. Indeed, having two (or more) tight box constraints  implies that   $x_1,x_2,x_3,x_4$ consists of two pairs of parallel vectors, which  contradicts the fact that ${\rm span}(x_1x_1^\top, x_2x_2^\top, x_3x_3^\top,x_4x_4^\top)=3$. In turn, the fact that we have at most one tight box constraint    implies that  at most one cycle inequality  
can be tight. For concreteness, say that 
$\ta_{14}=\ta_{13}+\ta_{23}+\ta_{24}$ and $\ta_{13}=\ta_{14}+\ta_{23}+\ta_{24}-2\pi$. Substituting the second equation into the first one we get that $\ta_{23}+\ta_{24}=\pi$, which when substituted back into the first equation gives that  $\ta_{14}=\ta_{13}+\pi$. In turn, using that $\ta_{13}, \ta_{14}\in [0,\pi],$ this implies that $\ta_{14}=\pi, \ta_{13}=0, $ i.e., we have two tight box constraints, a contradiction. Thus, it remains  to exhibit one tight cycle inequality. 

By assumption $C$ is extreme, and thus,  it admits a unique  PSD completion, i.e., there exists a unique choice for $c_{12}$ and $ c_{34}$ that makes the partial matrix \eqref{partial} PSD.  In particular, there exists a unique choice for the value of $\ta_{34}=\arccos (c_{34})$. 
We are now ready to conclude the proof of our theorem. 

We have already noted that under the assumptions of the theorem,  there exists a unique choice for the value of $\ta_{34}=\arccos (c_{34})$. Consequently, by Lemma \ref{cor:uniqueness}, the interval specified in \eqref{interval} should reduce to a single point, i.e., the two endpoints should coincide. This happens iff one expression  from the lower bound in \eqref{interval} is equal to another  expression  from the upper bound in \eqref{interval}. 

We  consider two cases. First, if these two inequalities have disjoint support,  we get a tight cycle inequality.   For example, from the equality $\ta_{31}-\ta_{41}=\ta_{32}+\ta_{42}$ we get the tight cycle inequality $\ta_{31}=\ta_{41}+\ta_{32}+\ta_{42}$. Second, if  the two inequalities have the same support, we get  a tight box inequality. In turn, this  gives a tight cycle inequality. For example, the equality $\ta_{32}-\ta_{42}=2\pi- (\ta_{32}+\ta_{42})$ gives the tight box inequality $\ta_{32}=\pi.$ As    $\rank(\hc)=2$, all  size three minors of $\hc$ are singular. In particular, the minor $\hc[2,3,4] $ is singular, and  thus,  by Theorem \ref{33case}, one of the triangle inequalities for $(2,3,4)$  is tight. Using that $\ta_{32}=\pi$,  combined with the fact that  we  can have at most one box inequality, we get  that $\ta_{34}+\ta_{24}=\pi$.  Lastly, 
the minor $\hc[1,3,4]$ is also  singular, and again by Theorem \ref{33case}, one of the triangle inequalities for $(1,3,4)$  is tight. For example, if $\ta_{13}+\ta_{14}=\ta_{34}$, by eliminating $\ta_{34}$ we get that $ \ta_{13}+\ta_{14}=\pi-\ta_{24}$, which is the tight cycle inequality  $\ta_{13}+\ta_{14}+\ta_{24}=\ta_{32}.$

Let us conclude with a remark that Part (ii) can be proved with a purely algebraic argument, i.e. never use the Gram vectors, by using the fact that rank of a matrix $M$ is the largest order of any non-zero minor in $M$. The advantage of this algebraic argument is to eliminate any kind of doubt caused by appeal to geometric intuitions.

\bibliography{citations}
\bibliographystyle{abbrv}
\end{document}